\documentclass[a4paper,UKenglish,cleveref,hyperref,autoref]{lipics-v2021}

\title{Answering Related Questions}
\titlerunning{Answering Related Questions}

\author{\'{E}douard Bonnet}{CNRS, ENS de Lyon, Université Claude Bernard Lyon 1, LIP UMR 5668, 69342 Lyon, France \and \url{http://perso.ens-lyon.fr/edouard.bonnet}}{edouard.bonnet@ens-lyon.fr}{https://orcid.org/0000-0002-1653-5822}{}

\authorrunning{\'E. Bonnet}

\Copyright{Édouard Bonnet}

\category{}

\relatedversion{}

\supplement{}

\funding{}

\acknowledgements{}

\nolinenumbers 

\hideLIPIcs  

\EventEditors{John Q. Open and Joan R. Access}
\EventNoEds{2}
\EventLongTitle{42nd Conference on Very Important Topics (CVIT 2016)}
\EventShortTitle{CVIT 2016}
\EventAcronym{CVIT}
\EventYear{2016}
\EventDate{December 24--27, 2016}
\EventLocation{Little Whinging, United Kingdom}
\EventLogo{}
\SeriesVolume{42}
\ArticleNo{23}

\usepackage[utf8]{inputenc}  

\usepackage[T1]{fontenc}
\usepackage{lmodern}
\usepackage[colorinlistoftodos,bordercolor=orange,backgroundcolor=orange!20,linecolor=orange,textsize=normalsize]{todonotes}
\usepackage{amsmath}  
\usepackage{amssymb}
\usepackage{bbm}
\usepackage{accents}
\usepackage{complexity}
\usepackage{booktabs}
\usepackage{bm}
\usepackage{paralist}
\usepackage{fixmath}
\usepackage{tcolorbox}

\makeatletter
\newtheorem*{rep@theorem}{\rep@title}
\newcommand{\newreptheorem}[2]{%
\newenvironment{rep#1}[1]{%
 \def\rep@title{#2 \ref{##1}}%
 \begin{rep@theorem}}%
 {\end{rep@theorem}}}
\makeatother

\newreptheorem{theorem}{Theorem}
\newreptheorem{lemma}{Lemma}

\usepackage{xspace}
\usepackage{tikz}
\usepackage[ruled,vlined,linesnumbered]{algorithm2e}

\usetikzlibrary{fit}
\usetikzlibrary{arrows}
\usetikzlibrary{patterns}
\usetikzlibrary{calc}
\usetikzlibrary{shapes}
\usetikzlibrary{positioning}
\usetikzlibrary{math}
\usetikzlibrary{shapes.symbols}
\usetikzlibrary{decorations.pathreplacing,calligraphy}
\usetikzlibrary{decorations.pathmorphing}
\usepackage[scr=boondox,scrscaled=1.05]{mathalfa}
\usetikzlibrary{shapes.geometric}

\newcommand{\defproblem}[3]{
  \vspace{1mm}
  \begin{tcolorbox}[
    colframe=black,        
    colback=white,         
    boxrule=0.5pt,         
    arc=4pt,               
    left=6pt, right=6pt,   
    top=6pt, bottom=6pt    
  ]    
    #1 \\
    {\bf{Input:}} #2 \\
    {\bf{Output:}} #3
  \end{tcolorbox}
  \vspace{1mm}
}

\newtheorem{question}{Question}

\newenvironment{proofofclaim}{\noindent \textsc{Proof:}}{\hfill$\Diamond$\medskip}

\crefname{claim}{Claim}{Claims}

\newcommand\dist{\mathsf{dist}\xspace}
\newcommand\diste{\mathsf{dist}_e\xspace}
\newcommand\distd{\mathsf{dist}_\Delta\xspace}

\newcommand\rel{\textsc{Sidestep}}
\newcommand\relp{\textsc{Sidestep-Pos}}
\newcommand\reln{\textsc{Sidestep-Neg}}

\newcommand\sol{\text{Sol}}

\newcommand\false{\mathsf{false}}
\newcommand\true{\mathsf{true}}
\newcommand\nil{\textsf{nil}}

\begin{document}

\maketitle

\begin{abstract}
  We introduce the meta-problem \rel$(\Pi, \dist, d)$ for a~problem $\Pi$, a~metric $\dist$ over its inputs, and a~map $d: \mathbb N \to \mathbb R_+ \cup \{\infty\}$.
  A~solution to \rel$(\Pi, \dist, d)$ on an input $I$ of $\Pi$ is a~pair $(J, \Pi(J))$ such that $\dist(I,J) \leqslant d(|I|)$ and $\Pi(J)$ is a~correct answer to $\Pi$ on input~$J$.
  This formalizes the notion of answering a~related question (or sidestepping the question), for which we give some motivations, and compare it with the adjacent concepts of smoothed analysis, certified algorithms, planted problems, modification problems, and approximation algorithms.
  Informally, we call hardness radius the ``largest'' $d$ such that \rel$(\Pi, \dist, d)$ is \NP-hard.
  This framework calls for establishing the hardness radius of problems~$\Pi$ of interest for the relevant distances $\dist$.

  We exemplify it with graph problems and two distances $\distd$ and $\diste$ (the edge edit distance) such that $\distd(G,H)$ (resp.~$\diste(G,H)$) is the maximum degree (resp.~number of edges) of the symmetric difference of $G$ and $H$ if these graphs are on the same vertex set, and $+\infty$ otherwise.
  Thus when solving \rel$(\Pi, \distd, d)$ (resp.~\rel$(\Pi, \diste, d)$) on an $n$-vertex input~$G$, acceptably close graphs $H$ are obtained by XORing $G$ with a~graph of maximum degree at~most~$d(n)$ (resp.~having at~most $d(n)$ edges).
  We show that the decision problems \textsc{Independent Set}, \textsc{Clique}, \textsc{Vertex Cover}, \textsc{Coloring}, \textsc{Clique Cover} have hardness radius $n^{\frac{1}{2}-o(1)}$ for $\distd$, and $n^{\frac{4}{3}-o(1)}$ for $\diste$, that \textsc{Hamiltonian Cycle} (or \textsc{Hamiltonian Path}) has hardness radius 0 for $\distd$, and somewhere between $n^{\frac{1}{2}-o(1)}$ and $n/3$ for $\diste$, and that \textsc{Dominating Set} has hardness radius $n^{1-o(1)}$ for $\diste$.
  We leave several open questions.
\end{abstract}

\section{Introduction}\label{sec:intro}

At the end of a~talk or in everyday life, we can be asked three kinds of well-formed questions: \emph{easy} ones, to which we can quickly provide an answer, \emph{hard} ones, which we cannot directly address but are still able to mention something related, and \emph{very hard} ones, for which even the latter is beyond us.
In this paper, we formalize the distinction between \emph{hard} and \emph{very hard} questions for computational problems.
To agree on the meaning of ``\emph{related}'' we need to specify a~metric on the input space of the problem at hand, as well as a~radius---possibly depending on the input size---within which an instance is deemed acceptably close to the initial question.
If the radius is the constant zero function, we are in fact solving the original problem.
At the other extreme, if the radius is infinite, then we shall simply find a~trivially solvable instance; an easy task for most problems.
It is then natural to determine how small this radius can be while keeping the meta-problem of ``sidestepping the question'' tractable. 
Informally, we call \emph{hardness radius} the supremum of the radii below which the problem is intractable. 

This paper serves as a~theoretical basis for \emph{answering related questions}, exemplified on central NP-complete graph problems, with the following motivations in mind.

\medskip

Efficiently solving the meta-problem for some task $\Pi$ (in the regime where the radius allows it) makes it possible to build a~customized dataset of graphs each labeled by an appropriate answer to~$\Pi$.
Indeed, if one wants the dataset on which $\Pi$ is solved to be some precise set $\mathcal D$ or to be sampled according to a~chosen distribution, one at least gets a~dataset $\mathcal D'$ that is pointwise close to the desired target.  
This can be useful to create benchmarks, curate instances for a~programming competition, or enable supervised learning for~$\Pi$.
Depending on the exact use case, we may ask for more properties from the instances of~$\mathcal D'$.
Ideally, solving $\Pi$ on the inputs of $\mathcal D'$ should remain as challenging as on~$\mathcal D$, or at least remain nontrivial, and as representative as possible of the original target distribution.
Regarding the last use case, while machine learning has become a~promising approach to combinatorial optimization~\cite{cappart2023gnnco}, supervised methods face a~clear bottleneck: producing labeled training instances requires solving possibly NP-hard problems.
Consequently, many recent works favor reinforcement learning~\cite{barrett2020eco} or unsupervised learning~\cite{wang2022principled,sun2022annealed,parkar2025localupdates}.
Applying supervised learning in this context relies on additional ad-hoc ingredients such as efficient problem-specific solvers \cite{SunY23}, leveraging sharp phase transitions and CSP solvers \cite{Lemos19}, or problem-dependent generation of hard instances~\cite{Lelarge25}.

The magnitude of the hardness radius may inform, at least for some relevant input metrics, about the degree of difficulty of a~problem.
For a~fixed distance between inputs of \NP-hard graph problems, this draws a~hierarchy of ``increasingly hard'' problems. 

\medskip

\textbf{\rel.}
Given a~problem~$\Pi$, a~metric $\dist$ over its inputs, and a~non-decreasing function $d: \mathbb N \to \mathbb R_+ \cup \{\infty\}$, we introduce the following problem.

\defproblem{\rel$(\Pi, \dist, d)$}{An input $I$ of $\Pi$.}{A~pair $(J,\Pi(J))$ with $\dist(I,J) \leqslant d(|I|)$.}

In the above, $|I|$ denotes the \emph{size} of input $I$, defined as a~natural number.
When the main part of $I$ is a~graph, $|I|$ denotes its number of vertices.
We further impose that the images of $\dist$ and $d$ can be computed in polynomial time in the size of their arguments.
Note that $\Pi$ can be a~\emph{decision} or a~\emph{function} problem.
And in both cases \rel$(\Pi, \dist, d)$ is itself a~function problem.
Thus $\Pi(J)$ lies in $\{\true, \false\}$ if~$\Pi$ is a~decision problem, and it is the correct solution to $\Pi$ on input $J$, or a special symbol $\nil$ indicating the absence of a~solution, when $\Pi$ is a~function problem.
Often, there may be several correct solutions.
We denote the set of all correct solutions to $\Pi$ on~$J$ by~$\sol_\Pi(J)$, and $\Pi(J)$ should be thought of as any fixed element of~$\sol_\Pi(J)$.
We make the slight abuse of notation of writing the expression $d(n)$ instead of the function $d$ in the third field of \rel.

In this paper, we use \emph{metric}, \emph{distance}, and \emph{distance function} as synonyms.
The maps we will consider $\diste$ and $\distd$ are indeed metrics.
However, the one statement where the distance functions are unspecified, \cref{thm:lpi}, does not require any metric axioms. 
In general, one can perfectly well define and study \rel\ with a~mere ``proximity measure.''

\medskip

\textbf{Considered metrics.}
We will exclusively deal with problems $\Pi$ whose input is an unweighted undirected graph, and optionally an~integer threshold, and consider two distances between graphs.
Both require that two graphs have the same vertex set for the distance between them to be finite.
For every two graphs $G, H$, $\diste(G,H) := \infty$ if $V(G) \neq V(H)$ and $$\diste(G,H) := |E(G) \triangle E(H)|,$$ otherwise, where~$\triangle$ denotes the symmetric difference.
For two graphs $G, H$ on the same vertex set, $V(G)$, we denote by $G \triangle H$ the graph with vertex set $V(G)$ and edge set $E(G) \triangle E(H)$.
We set $\distd(G,H) := \infty$ if $V(G) \neq V(H)$ and $$\distd(G,H) := \Delta(G \triangle H),$$ otherwise, where $\Delta(\cdot)$ denotes the maximum degree of its graph argument.
If the whole input consists of a~graph together with some additional integral, rational, or real parameters (typically serving as thresholds for a~graph property), a~finite distance between $(G,\overline k)$ and $(H,\overline \ell)$ implies that $\overline k = \overline \ell$.
Specifically, for every considered distance~$\dist$, $\dist((G,\overline k),(H,\overline \ell)) = \dist(G,H)$ if $\overline k = \overline \ell$, and $\dist((G,\overline k),(H,\overline \ell)) =\infty$ otherwise.
In particular, we will often exclude the parameters (when they exist at all) from the expressed distances.
 
In this paper, we chose to limit the distances to $\diste$ and $\distd$, which we believe are the two main distances over \emph{labeled} graphs.
Since we require a~polynomial-time algorithm to evaluate the distances, we cannot really afford a~metric invariant under isomorphism (because no polynomial algorithm is known for \textsc{Graph Isomorphism}); hence the \emph{labeled} vertex sets.
Then, $\diste$ essentially is the Hamming distance over the edge characteristic vector, and $\distd$ is the \emph{per-vertex} counterpart of $\diste$.
For example, we find natural and interesting the question of the minimum $d$ for which \rel$($\textsc{Ha\-mil\-tonian Cycle}$, \diste, d)$ is polynomial-time solvable, that is, in plain English: How many edge edits are necessary and sufficient to efficiently determine the Hamiltonicity of a~graph?
We find it appealing that the framework encompasses many questions of this kind that have not been explored in this precise form.

\medskip

\textbf{Hardness and tractability radii.}
A~function problem is said \emph{\NP-hard} if every problem in \NP~reduces to it via a~polynomial-time Turing reduction. 
The \emph{hardness radius} of a~problem $\Pi$ for (or \emph{under}) distance $\dist$ is at~least some non-negative, non-decreasing function $d$ if \rel$(\Pi, \dist, d)$ is~\NP-hard.
Similarly, the \emph{tractability radius} of~$\Pi$ for distance $\dist$ is at~most some non-negative, non-decreasing function $d$ if \rel$(\Pi, \dist, d)$ is~in~\FP, the polynomial-time class for function problems.
As the hardness and tractability radii are \emph{functions}, formalizing what they are \emph{equal to} through infima and suprema would require a~total order over maps. 
However, two non-decreasing functions $f, g: \mathbb N \to \mathbb R_+$ can be asymptotically incomparable, in the sense that for every $n \in \mathbb N$ there are $n_f, n_g \geqslant n$ such that $f(n_f) > g(n_f)$ and $f(n_g) < g(n_g)$.
Rather than defining the family of allowed maps, we will leave the definition as is: mere lower and upper bounds.

In some situations, though, we will be able to ``pinpoint'' more or less precisely the hardness or tractability radius.
We may say that $\Pi$ has for distance $\dist$ \emph{hardness radius}
\begin{compactitem}
\item $\gamma \in \mathbb R$ if \rel$(\Pi, \dist, \gamma)$ is \NP-hard, and for all $\varepsilon > 0$ such that $\gamma+\varepsilon$ is in the image of~$\dist$, \rel$(\Pi, \dist,   \gamma + \varepsilon)$ is in \FP,
  \item $n^{\gamma-o(1)}$ if \rel$(\Pi, \dist, n^\gamma)$ is in \FP, and for all $\varepsilon \in (0,\gamma]$, \rel$(\Pi, \dist, n^{\gamma-\varepsilon})$ is \NP-hard.
\end{compactitem}
Note that the hardness and tractability radii can in principle be separated by a~``zone'' of \NP-intermediate problems.
While our main distances, $\distd$ and $\diste$, only take integer values, and we could have given simplified definitions for them, the above supports the general case when the metric is valued in~$\mathbb R_+ \cup \{\infty\}$.

\subsection{Related concepts}

Here we detail how our meta-problem relates to and differs from well-established concepts.

\medskip

\textbf{Smoothed analysis.}
Introduced by Spielman and Teng to explain the practical effectiveness of the simplex algorithm despite its worst-case exponential runtime~\cite{SpielmanT04}, smoothed analysis has been successful in serving a~similar purpose for several algorithms operating over a~continuous input space, 
 such as the $k$-means method~\cite{Arthur11} or the 2-opt heuristic for the Euclidean \textsc{TSP}~\cite{Englert16}; also see~\cite{SpielmanT09}.
The (worst-case) complexity of~\rel$(\Pi, \dist, d)$ can be seen as an adversary's choosing the worst input $I$, and our choosing the most convenient input $I'$ at distance at~most~$d(|I|)$ from $I$, and then solving $\Pi$ on input~$I'$.
Smoothed analysis can be thought of as following the above process with the twist that $I'$ is instead chosen by Nature, according to some probability distribution. 
We refer the interested reader to \cite{SpielmanT04,SpielmanT09} for formal definitions.

\begin{figure}[h!]
\centering

\begin{subfigure}[t]{0.45\textwidth}
\centering
\begin{tikzpicture}[scale=1.2]
    \def\Xmin{-2.5}
    \def\Xmax{2.5}
    \def\Ymin{-2.5}
    \def\Ymax{2.5}

    \draw[rounded corners, thick, gray] (\Xmin,\Ymin) rectangle (\Xmax,\Ymax);
    \node[gray] at (\Xmin+1.3,\Ymax-0.3) {Input space of $\Pi$};
    
    \coordinate (I) at (-0.5,0.5);
    \def\Radius{1}
    
    \shade[inner color=blue!70, outer color=blue!10] (I) circle (\Radius cm);
    
    \node[fill=red,circle,inner sep=-0.07cm] (nI) at (I) {} ;
    \node[below] at ($(I)+(0,-0.08cm)$) {\textcolor{red}{$I$}};
    
    \foreach \angle/\radius in {20/0.6, 55/0.8, 90/0.5, 125/0.9, 160/0.4, 200/0.7, 235/0.6, 275/0.8, 320/0.5}
    {
        \coordinate (P) at ($(I)+(\angle:\radius cm)$);
        \filldraw[black] (P) circle (1pt);
    }
    
    \coordinate (Perturbed) at ($(I)+(55:0.8cm)$);
    \draw[->, thick, blue, decorate, decoration={snake, amplitude=1.5}] (nI) -- (Perturbed) node[midway, right] {};

    \node[below] at ($(I)+(0,-1.05cm)$) {perturbation support};

\end{tikzpicture}
\caption{Smoothed analysis.}
\end{subfigure}
\hfill
\begin{subfigure}[t]{0.45\textwidth}
\centering
\begin{tikzpicture}[scale=1.2]
    \def\Xmin{-2.5}
    \def\Xmax{2.5}
    \def\Ymin{-2.5}
    \def\Ymax{2.5}

    \draw[rounded corners, thick, gray] (\Xmin,\Ymin) rectangle (\Xmax,\Ymax);
    \node[gray] at (\Xmin+1.3,\Ymax-0.3) {Input space of $\Pi$};
    
    \coordinate (I) at (0.5,-0.5);
    \def\Radius{1}
    
    \filldraw[fill=green!10, draw=green!50] (I) circle (\Radius cm);
    
    \node[fill=red,circle,inner sep=-0.07cm] (nI) at (I) {} ;
    \node[above] at ($(I)+(0,0.04cm)$) {\textcolor{red}{$I$}};
    
    \foreach \angle/\radius in {15/0.7, 60/0.5, 100/0.8, 145/0.6, 185/0.9, 225/0.6, 265/0.8, 305/0.5}
    {
        \coordinate (P) at ($(I)+(\angle:\radius cm)$);
        \filldraw[black] (P) circle (1pt);
    }
    
    \coordinate (I') at ($(I)+(45:0.7cm)$);
    \filldraw[black] (I') circle (1pt) node[right] {\textcolor{green!60!black}{$I'$}};
    
    \draw[->, thick, green!60!black] (nI) -- (I') node[midway, above left] {};
    
    \coordinate (BoundaryPoint) at ($(I)+(250:1cm)$);
    \draw[<->, thick, black] (nI) -- (BoundaryPoint) node[midway, right] {$d$};
    
    \node[below] at ($(I)+(0,-1.05cm)$) {nearby instances};

\end{tikzpicture}
\caption{\rel$(\Pi, \dist, d)$.}
\end{subfigure}

\caption{Comparison of the smoothed analysis of $\Pi$ and \rel$(\Pi, \dist, d)$.
  The \emph{moves} of the adversary, Nature, and ours are depicted in red, blue, and green, respectively.}
\label{fig:smoothed-vs-sidestep}
\end{figure}
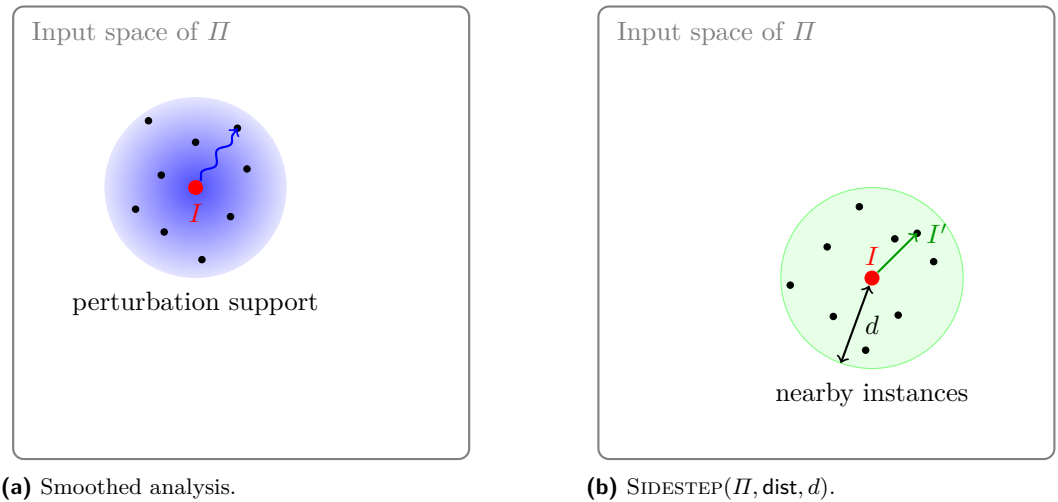

While the aim of smoothed analysis is to model real data, and thus to better describe the practical performance of an algorithm, the main motivation of solving \rel$(\Pi, \dist, d)$ lies elsewhere.
However, both endeavors are comparable in that they hierarchize \NP-hard problems by a~finer-grained complexity notion departing from worst-case analysis.
In fact, smoothed analysis also has its own ``hardness radius'' in the supremum of the variance for which the average complexity around the worst instance~$I$ is at~least superpolynomial.\footnote{And \emph{smoothed polynomial time} is reached when below the ``hardness radius'' the variance is low enough that the \emph{allowed time} exceeds that of the brute-force approach.}

  \medskip

  \textbf{Certified algorithms.}
  Perturbation stability~\cite{Bilu12} (also called \emph{instance stability}, \emph{Bilu--Linial stability}, or \emph{perturbation resilience}) defines \emph{$\gamma$-stable instances} of edge-weighted graph problems as those for which a~\emph{$\gamma$-perturbation}, i.e., an independent rescaling of the edge weights by a~multiplicative factor in $[1,\gamma]$, cannot change the optimum solution (only its cost).
  Under some slight adjustments, the NP-hard \textsc{Max Cut} problem can be solved in polynomial-time on $\gamma$-stable instances with sufficiently large $\gamma$~\cite{Bilu12}.
  The theory of perturbation stability has been exported to other graph problems (see for instance~\cite{Awasthi12}), and has given rise to the notion of \emph{certified algorithms} by Makarychev and Makarychev~\cite{Makarychev20}.

  For an edge-weighted graph problem $\Pi$, a~\emph{$\gamma$-certified algorithm} is one that, on any input~$I$, outputs an optimum solution to a~$\gamma$-perturbation of~$I$.
  This is a~particular case of our framework when we do not impose $\dist$ to be symmetric: \rel$(\Pi,\dist,\gamma)$ where $\dist(I,I')$ is equal to the largest ratio between an edge weight of $I'$ and the weight of the same edge in~$I$, and to $\infty $ if $I$ and $I'$ have not the same vertex and edge sets or if this ratio is smaller than~1.
  As our focus is on unweighted purely discrete problems, the results in this paper are disjoint from the literature on certified algorithms.  
  
\medskip

\textbf{Planted problems.}
Planted problems, and their main representative, the planted clique problem~\cite{Jerrum92}, ask to distinguish between a~(plain) random graph (usually $\mathcal G(n,1/2)$) and the ``same'' random graph augmented by some edges forming a~solution to an \NP-hard problem; be it by adding all the edges on a~selected vertex subset (\textsc{Clique}) or the edges of a~spanning cycle (\textsc{Hamiltonian Cycle}).
The search variant consists of actually reporting this planted solution.
In some way, \rel$(\Pi, \dist, d)$ is a~\emph{planting} problem instead.
It is about editing the inputs with features that make them---at~least to the \emph{editor}---easily identifiable as positive or as negative instances.



\medskip

\textbf{Edge modification problems.}
Edge modification problems ask, given a~graph $G$ and a~non-negative integer~$k$, whether at~most $k$~edges of $G$ can be edited (i.e., added or removed) for the resulting graph to belong to a~particular graph class~$\mathcal C$.
Most often, the parameterized complexity of these problems with respect to~$k$ is investigated.
There is a~recent survey on the topic~\cite{Crespelle23}.

Conceptually, this is close to \rel$(\Pi_{\mathcal C}, \diste, k)$, where $\Pi_{\mathcal C}$ is the recognition problem of class~$\mathcal C$.
Let us highlight the key differences.
The most obvious one is that, unlike in the edge modification problem to~$\mathcal C$, instances of~\rel$(\Pi_{\mathcal C}, \diste, k)$ can be addressed by creating a~negative instance.
Another important difference is that edge modification problems almost exclusively focus on classes~$\mathcal C$ that are easy to recognize, such as forests, interval graphs, or cographs.
Indeed, if the problem is already \NP-hard when $k=0$, there is not much to be done.
On the contrary, in \rel$(\Pi_{\mathcal C}, \diste, k)$, $k$ is part of the problem specification, \emph{not} the input, and we try and establish the smallest $k$ such that a~polynomial-time algorithm exists.
Therefore in our case, $\Pi_{\mathcal C}$ is intended to be hard.

\medskip

\textbf{Approximation algorithms.}
Some lower bounds in this paper are derived from known hardness-of-approximation results.
However, \rel\ is only loosely related to approximation.
Approximation keeps the input fixed and relaxes the output objective, whereas \rel\ allows changing the input but requires an exact answer on the modified instance.
Switching between yes- and no-instances also has no analogue in approximation algorithms.
We finally note that \cref{thm:ham-hardness-diste,thm:ds-hardness-diste} are \emph{not} obtained via inapproximability. 

\subsection{Our results}

In~\cref{sec:ham} we illustrate the new framework with the hamiltonicity problem. 
We first observe that \rel$(\textsc{Hamiltonian Cycle}, \distd, 1)$ is tractable, while \rel$(\textsc{Hamiltonian}$ $\textsc{Cycle}, \distd, 0)$ is not, since it coincides with \textsc{Hamiltonian Cycle}.  

\begin{theorem}\label{thm:ham-hr-distd}
  \textsc{Hamiltonian Cycle} has hardness radius 0 under~$\distd$. 
\end{theorem}

To establish the hardness radius of~\textsc{Hamiltonian Cycle}, $\diste$ proves to be a~more challenging distance.
We improve the tractability radius from $n/2$ (due to a~simple perfect matching argument used to obtain~\cref{thm:ham-hr-distd}) to~$n/3$.

\begin{theorem}\label{thm:ham-tr-diste}
  \rel$(\textsc{Hamiltonian Cycle}, \diste, n/3)$ is in \FP.
\end{theorem}

The algorithm of~\cref{thm:ham-tr-diste} uses a~path-growing argument reminiscent of some proofs of Dirac's theorem~\cite{Dirac52} or Ore's theorem~\cite{Ore60}. 
We do not manage to match~\cref{thm:ham-tr-diste} with a~tight lower bound.
However we show that the hardness radius of \textsc{Hamiltonian Cycle} under $\diste$ is at~least $n^{\frac{1}{2}-o(1)}$.

\begin{theorem}\label{thm:ham-hardness-diste}
  For any $\beta > 0$, \rel$($\textsc{Hamiltonian Cycle}$, \diste, n^{\frac{1}{2}-\beta})$ is \NP-hard.
\end{theorem}

\Cref{thm:ham-hardness-diste} offers a~first example of what we call a~\emph{robust reduction}, i.e., in the particular case of the metric $\diste$, a~reduction whose positive (resp. negative) instances can withstand the edition of some number of edges (in that case $n^{\frac{1}{2}-\beta}$).
In the proof of~\cref{thm:ham-hardness-diste}, we indeed design a~reduction where the produced Hamiltonian graphs cannot be made non-Hamiltonian by removing at~most~$n^{\frac{1}{2}-\beta}$ edges, and the produced non-Hamiltonian graphs cannot be made Hamiltonian by adding at~most~$n^{\frac{1}{2}-\beta}$ edges.
The former is ensured by applying Ore's theorem on some specific induced subgraphs, while the latter is arranged by piecing together copies of the same 2-connected component.

We further give some evidence (see~\cref{thm:barrier-longest-path}) that lowering the tractability radius at~$n^{1-\beta}$ for some fixed $\beta > 0$ may be challenging since that would imply beating the current polynomial-time best approximation algorithm for \textsc{Longest Path} in Hamiltonian graphs.

Turning our attention to the \textsc{Dominating Set} problem, we prove the following.
\begin{theorem}\label{thm:ds-hr-diste}
  \textsc{Dominating Set} has hardness radius $n^{1-o(1)}$ for $\diste$.
\end{theorem}
The upper bound of $n-1$ for the tractability radius of \textsc{Dominating Set} under $\diste$ is trivial. 
We give a~simple argument to bring this upper bound down to $n/e$, where $e$ is Euler's number.   
\medskip

\Cref{sec:is} is devoted to establishing the hardness (or tractability) radii for $\distd$ and $\diste$ of five central graph problems. 

\begin{theorem}\label{thm:1/2-4/3}
  \textsc{Independent Set}, \textsc{Clique}, \textsc{Vertex Cover}, \textsc{Coloring}, and \textsc{Clique Cover} have hardness radius $n^{\frac{1}{2}-o(1)}$ under $\distd$ and hardness radius $n^{\frac{4}{3}-o(1)}$ under $\diste$. 
\end{theorem}

\Cref{thm:1/2-4/3} comprises twenty substatements: two algorithms and two hardness results for five problems.
Let us first note that all the problems of~\cref{thm:1/2-4/3} should be thought of as decision problems.
They take a~graph and an integer threshold, and ask for a~feasible solution at that threshold; see~\cref{sec:is} for the definitions.
However, since \cref{thm:1/2-4/3} even works when a~feasible solution is required in output (not just a~yes/no answer), we deal with this \emph{function variant} of the decision problem.

Solving \rel$($\textsc{Max Independent Set}$, \dist, d)$, where \textsc{Max Independent Set} is the optimization variant, is in principle much harder (see~\cref{sec:open}) than solving \rel$($\textsc{Independent Set}$, \dist, d)$.
In the latter, one can edit the input graph to be clearly below or clearly above the threshold.
In the former, whichever nearby graph ends up being selected for the output, its independence number should be precisely known.

We start by establishing \cref{thm:1/2-4/3} for~\textsc{Independent Set}.
The algorithms are simple win--win arguments, planting witnesses of low or of high independence number.
The matching lower bounds are based on the famously high inapproximability of~\textsc{Max Independent Set}~\cite{Hastad96,Zuckerman07}.
We then show that a~restricted form of polynomial-time reductions, namely polynomial-time isomorphisms both preserving the input size and the distance between inputs, transfers the hardness radius; see~\cref{thm:lpi} for a~precise statement.
As a~consequence, we obtain~\cref{thm:1/2-4/3} for \textsc{Clique} and \textsc{Vertex Cover}.

The algorithms for~\textsc{Independent Set} almost readily work for \textsc{Clique Cover}, and hence for \textsc{Coloring} via~\cref{thm:lpi}.
We finally adapt the hardness results for \textsc{Coloring} (and hence for \textsc{Clique Cover} via~\cref{thm:lpi}).
This requires a~bit more work, especially for the distance $\diste$ where we use a~classical upper bound of the chromatic number of $m$-edge graphs (\cref{obs:m-to-degen}) and the Cauchy--Schwarz inequality.

\begin{table}[t]
  \centering
  \small
  \begin{tabular}{lccc}
    \hline
    Problem & Easy at radius & Hard at radius & Statement \\
    \hline
    \multicolumn{4}{c}{for $\diste$} \\
    \hline
    \textsc{Hamiltonian Cycle}
      & $n/3$
      & $n^{1/2-o(1)}$
      & \cref{thm:ham-tr-diste,thm:ham-hardness-diste} \\ 

    \textsc{Dominating Set}
      & $n/e$
      & $n^{1-o(1)}$
      & \cref{thm:ds-algo-diste,thm:ds-hardness-diste} \\

    \textsc{Independent Set}
      & $n^{4/3}$
      & $n^{4/3-o(1)}$
      & \cref{thm:algo-is-diste,thm:hardness-is-diste} \\

    \textsc{Clique}, \textsc{Vertex Cover}
      & $n^{4/3}$
      & $n^{4/3-o(1)}$
      & \cref{thm:clique-vc} \\

    \textsc{Coloring}, \textsc{Clique Cover}
      & $n^{4/3}$
      & $n^{4/3-o(1)}$
      & \cref{thm:algo-coloring,thm:hardness-coloring-diste} \\
    \hline
    \multicolumn{4}{c}{for $\distd$} \\
    \hline
    \textsc{Hamiltonian Cycle}
      & $1$
      & $0$
      & \cref{thm:ham-hr-distd} \\ 

    \textsc{Independent Set}
      & $n^{1/2}$
      & $n^{1/2-o(1)}$
      & \cref{thm:algo-is,thm:hardness-is} \\

    \textsc{Clique}, \textsc{Vertex Cover}
      & $n^{1/2}$
      & $n^{1/2-o(1)}$
      & \cref{thm:clique-vc} \\

    \textsc{Coloring}, \textsc{Clique Cover}
      & $n^{1/2}$
      & $n^{1/2-o(1)}$
      & \cref{thm:algo-coloring,thm:hardness-coloring} \\
    \hline
  \end{tabular}
  \caption{Summary of the results.}
  \label{tab:summary-results}
\end{table}

\medskip

As far as approximation algorithms and parameterized complexity are concerned, \textsc{Vertex Cover}, which admits fixed-parameter tractable algorithms and polynomial-time \mbox{2-approximation} algorithms, is  ``much more tractable'' than the other problems of~\cref{thm:1/2-4/3}.
By contrast, the complexity measure of ascending hardness radius (for, say, $\diste$) groups \textsc{Vertex Cover} together with \textsc{Independent Set} and \textsc{Coloring}, higher up than \textsc{Dominating Set} and \textsc{Hamiltonian Cycle}.

\subsection{Open questions and potential future work}\label{sec:open}

\Cref{sec:ham} suggests the following question.

\begin{question}\label{q:ham-hr-diste}
  What is the hardness radius of \textsc{Hamiltonian Cycle} under $\diste$?
\end{question}

Toward answering~\cref{q:ham-hr-diste}, is it true that for any positive $k$, \rel$($\textsc{Hamil}$\- $\textsc{tonian Cycle}$, \diste, n/k)$ is in \FP?
We ask the same question for \textsc{Dominating Set}.

\begin{question}\label{q:ds-tr-diste}
  For any positive $k$, is \rel$(\textsc{Dominating Set}, \diste, n/k)$ in \FP?
\end{question}

In~\cref{thm:ds-algo-diste}, we show~\cref{q:ds-tr-diste} for $k = e$.

\begin{question}\label{q:ds-hr}
  What is the hardness radius of \textsc{Dominating Set} under $\distd$?
\end{question}

It is easy to see that for every integer $s > 0$, \rel$(\textsc{Dominating Set}, \distd, n/s)$ is in \FP.
On inputs $(G,k)$ with $k \leqslant s$, use the brute-force approach to decide if $G$ has a~dominating set of size $k$.
When instead $k > s$, pick any $k$ vertices and make each adjacent to $n/k < n/s$ distinct vertices; the original $k$ vertices now form a~dominating set. 

\begin{question}\label{q:max-clique}
  What is the hardness radius of \textsc{Max Clique} under $\distd$ and under~$\diste$?
\end{question}
More generally, we ask~\cref{q:max-clique} for the optimization versions of the problems of~\cref{thm:1/2-4/3}.

The \NP-hard \textsc{Triangle Partition} problem inputs an $n$-vertex graph $G$ with $n$ divisible by~3, and asks for a~partition of $V(G)$ into $n/3$ triangles.
It is straightforward that \rel$(\textsc{Triangle}$ $\textsc{Partition}, \distd, 2)$ is in \FP.

\begin{question}\label{q:triangle-partition}
  Is \rel$(\textsc{Triangle Partition}, \distd, 1)$ in \FP?
\end{question}

It is easy to see that \rel$(\textsc{Triangle Partition}, \diste, 2n/3)$ is in \FP.
How much can the tractability radius be pushed?

\begin{question}\label{q:triangle-partition-diste}
  What is the hardness radius of \textsc{Triangle Partition} under $\diste$?
\end{question}

For some applications when $\Pi$ is a~decision problem, one may prefer the one-sided variants of \rel$(\Pi, \dist, d)$, where one forces nearby instances to always be positive or always be negative.

\defproblem{\relp$(\Pi, \dist, d)$}{An input $I$ of $\Pi$.}{A~positive instance $J$ for $\Pi$ with $\dist(I,J) \leqslant d(|I|)$.}

\defproblem{\reln$(\Pi, \dist, d)$}{An input $I$ of $\Pi$.}{A~negative instance $J$ for $\Pi$ with $\dist(I,J) \leqslant d(|I|)$.}

These problems are one step closer to editing problems, especially \relp$(\Pi, \dist, d)$, but still the settings are somewhat distinct.
One can ask for hardness/tractability radii for these one-sided variants.

Although we only consider graph problems in the current paper, one can tackle the hardness/tractability radii of \rel$(\Pi, \dist, d)$ for formula, hypergraph, string, or geometric problems (to name a~few) with the relevant distances $\dist$ between inputs.
For instance, one could consider Hamming or edit distance for strings or formulas, or suitable perturbation measures for geometric inputs. 
Besides, the chosen dividing line to define these radii need not be \FP\,vs \NP-hard. 

Beyond the specific questions we raised, we would find it particularly interesting to develop a more systematic complexity theory for \rel\ (in the spirit of~\cref{thm:lpi}), identify reduction notions that preserve hardness and tractability radii, understand the one-sided and optimization variants, and determine whether broad classes of problems admit common explanations for their radii.

\section{Graph and set notation}\label{sec:prelim}

If $i$ and $j$ are two integers, we denote by $[i,j]$ the set of integers that are at~least~$i$ and at~most~$j$.
For every integer $i$, $[i]$ is a shorthand for $[1,i]$.

We denote by $V(G)$ and $E(G)$ the vertex set and the edge set, respectively, of a graph~$G$.
If $G$ is a~graph and $S \subseteq V(G)$, we denote by $G[S]$ the subgraph of $G$ induced by~$S$, and use $G-S$ as a~shorthand for $G[V(G) \setminus S]$.
We denote the open and closed neighborhoods of a vertex $v$ in $G$ by $N_G(v)$ and $N_G[v]$, respectively.
For $S \subseteq V(G)$, we set $N_G(S) := \bigcup_{v \in S}N_G(v) \setminus S$ and $N_G[S] := N_G(S) \cup S$.
A~\emph{dominating set} of $G$ is a~subset $S \subseteq V(G)$ such that $N_G[S]=V(G)$.
We say that $v \in V(G)$ (resp.~$S \subseteq V(G)$) dominates $X \subseteq V(G)$ if $X \subseteq N_G[v]$ (resp.~$X \subseteq N_G[S]$).
In every notation with a~graph subscript, we may omit it if the graph is clear from the context.

If $e \in E(G)$, we denote by $G-e$ the graph $G$ with edge $e$ removed (but the endpoints of $e$ remain).
More generally, if $F \subseteq E(G)$, $G-F$ is the graph obtained from $G$ by removing all the edges of~$F$ (but not their endpoints).
We denote by $G+e$ (resp.~$G+F$) the graph with vertex set $V(G)$ and edge set $E(G) \cup \{e\}$ (resp. $E(G) \cup F$).

The \emph{clique number} of a~graph $G$, denoted by~$\omega(G)$, is the maximum size of a~clique of~$G$, i.e., a~set of pairwise adjacent vertices.
The \emph{independence number} of $G$, denoted by~$\alpha(G)$, can be defined as the clique number of $\overline G$, \emph{the complement of $G$}, which flips edges and non-edges.
The \emph{chromatic number} of $G$, denoted by~$\chi(G)$, is the least number of colors required to properly color~$G$, i.e., give a~color to each vertex such that no edge has its two endpoints of the same color, or equivalently the least number of parts in a~partition of~$V(G)$ into independent sets of~$G$.
We will often use the simple fact that any graph $G$ of maximum degree $\Delta$ has chromatic number at~most~$\Delta+1$, and has an independent set of size at~least $|V(G)|/(\Delta+1)$. 

A~\emph{vertex cover} of~$G$ is a~subset $S \subseteq V(G)$ such that every edge of~$G$ has at~least one endpoint in~$S$.
A~\emph{Hamiltonian cycle} (resp.~\emph{Hamiltonian path}) is a~cycle (resp.~path) passing through each vertex exactly once.
A~graph is said \emph{Hamiltonian} if it admits a~Hamiltonian cycle.
A~\emph{perfect matching} of $G$ is a~matching of~$G$ with $\lfloor |V(G)|/2 \rfloor$ edges.

We denote by $A \triangle B$ the \emph{symmetric difference} of $A$ and $B$, i.e., $(A \setminus B) \cup (B \setminus A)$. 
If two graphs $G, H$ have the same vertex set, then $G \triangle H$ denotes the graph with vertex set $V(G)=V(H)$ and edge set $E(G) \triangle E(H)$.
A~useful simple fact about the symmetric difference is that $A = B \triangle C$ implies $A \triangle B = C$.
This holds when $A, B, C$ are sets or graphs.

\section{Hamiltonicity and Domination}\label{sec:ham}

We treat \textsc{Hamiltonian Cycle} (and similarly \textsc{Hamiltonian Path}) as a~function problem.
In particular, an algorithm for \rel$($\textsc{Hamiltonian Cycle}$, \dist, d)$ returns, on input $G$, a~pair $(G',C)$ where $C$ is a~Hamiltonian cycle of $G'$, or $(G',\nil)$ if $G'$ is not Hamiltonian. 

We first observe that the hardness radius of \textsc{Hamiltonian Cycle} under $\distd$ is 0.
(The same holds with some small adjustments for \textsc{Hamiltonian Path}.)
As the distance $\distd$ only takes integer values, we shall simply show that \rel$($\textsc{Hamiltonian Cycle}$, \distd, 1)$ is tractable.
Indeed, \rel$(\Pi, \dist, 0)$ is equivalent to $\Pi$ for every problem $\Pi$ and distance $\dist$; hence \rel$($\textsc{Hamiltonian Cycle}$, \distd, 0)$ is \NP-hard.
Thus the following theorem implies \cref{thm:ham-hr-distd}.

\begin{theorem}\label{thm:ham-cycle-easy-dd}
  \rel$($\textsc{Hamiltonian Cycle}$, \distd, 1)$ is in \FP.
\end{theorem}
\begin{proof}
  If $G$ is disconnected, simply output $(G,\nil)$ as a~disconnected graph cannot be Hamiltonian.
  We now assume that $G$ is connected.
  Compute a~perfect matching of~$G$ in polynomial time.
  If this step reports that $G$ has no perfect matching, again output $(G,\nil)$.
  Indeed a~Hamiltonian cycle yields a~perfect matching by taking every other edge on the cycle.

  We now assume that $G$ is connected, and that we have found a~perfect matching $M := \{u_1v_1, u_2v_2, \ldots, u_pv_p\}$ of~$G$.
  If $V(G)$ is odd, we further assume without loss of generality (as $G$ is connected) that the unique vertex $u$ of $G$ outside $M$ is adjacent to~$u_1$.
  We build the following graph $J$: Start with the edgeless graph on vertex set $V(G)$, and for every $i \in [p-1]$ if $v_iu_{i+1} \notin E(G)$ then add the edge $v_iu_{i+1}$ to $E(J)$.
  Finally if $V(G)$ is odd and $v_pu \notin E(G)$, add $v_pu$ to $E(J)$.
  If $V(G)$ is even and $v_pu_1 \notin E(G)$, add $v_pu_1$ to $E(J)$.

  We set $H := G \triangle J$, thus $G \triangle H = J$.
  It can be observed that $J$ has maximum degree at~most~1.
  Therefore $\distd(G,H) \leqslant 1$.
  And we can output $(H,(u)u_1v_1u_2v_2 \ldots u_pv_p)$.
\end{proof}

Note that the proof of~\cref{thm:ham-cycle-easy-dd} shows that \rel$($\textsc{Hamiltonian Cycle}$, \diste, n/2)$ is in~\FP.
We now see that we can improve on the $n/2$ bound.

\begin{reptheorem}{thm:ham-tr-diste}\label[theorem]{thm:ham-cycle-easy-n/3}
  \rel$($\textsc{Hamiltonian Cycle}$, \diste, n/3)$ is in \FP.
\end{reptheorem}
\begin{proof}
  Greedily grow a~path $P$ in the $n$-vertex input $G$.
  That is, while an endpoint of $P$ has a~neighbor outside $V(P)$, add this neighbor to the path.
  When this process ends, we get a~path $P=v_1v_2 \ldots v_h$.
  If $h=n$, we output $(G+v_1v_h, v_1v_2 \ldots v_h)$.
  Similarly, if $|V(G) \setminus V(P)|=n-h$ is smaller than $n/3$, it is easy to add at~most $n/3$ edges to $G$ to make it Hamiltonian.
  We therefore assume that $h \leqslant 2n/3$, and further assume, as before, that $G$ is connected.

  We will now show that either at~least one endpoint $v_1, v_h$ has degree at~most $h/2 \leqslant n/3$, or we can find a~strictly longer path $P'$ and restart the argument.
  First, we observe that $v_1$ and $v_h$ have, by construction, no neighbors in $V(G) \setminus V(P)$ (otherwise $P$ could be extended).
  \begin{claim}\label{clm:longer-path}
    If both $v_1v_{i+1}, v_iv_h$ are edges of~$G$, then one can effectively find in $G$ a~path on $h+1$ vertices.
  \end{claim}
  \begin{proofofclaim}
    As $G$ is connected (and $P$ is not a~Hamiltonian path), there is a~vertex $v_j \in V(P)$ with a~neighbor $v \in V(G) \setminus V(P)$.
    By symmetry, we can assume that $1 \leqslant j \leqslant i$.
    Then, $vv_jv_{j+1} \ldots v_{i-1}v_iv_hv_{h-1} \ldots v_{i+2}v_{i+1}v_1v_2 \ldots v_{j-2}v_{j-1}$ is the desired $P'$ on $h+1$ vertices. 
  \end{proofofclaim}
  
  While \cref{clm:longer-path} applies, we proceed with the longer path $P'$ (this will happen fewer than $n$ times).
  If it does not apply, then the sum of the degrees of $v_1$ and $v_h$ is upper bounded by~$h$.
  Hence one of $v_1, v_h$ has degree at~most $h/2 \leqslant n/3$.
  In this case, we isolate $v_1$ or $v_h$ by at~most $n/3$ edge deletions, thereby producing a~clearly negative instance.
\end{proof}

We next see how to make robust reductions for \textsc{Hamiltonian Cycle}, which can essentially sustain a~square root number of edge edits (in the number of vertices).
There is still some room between this lower bound and the upper bound of~\cref{thm:ham-cycle-easy-n/3}.
We will later explain why a~truly sublinear upper bound might be difficult to obtain.

\begin{reptheorem}{thm:ham-hardness-diste}\label[theorem]{thm:ham-cycle-hard}
  For any $\beta > 0$, \rel$($\textsc{Hamiltonian Cycle}$, \diste, n^{\frac{1}{2}-\beta})$ is \NP-hard.
\end{reptheorem}
\begin{proof}
  We reduce from \textsc{Hamiltonian Path} on subcubic graphs with two vertices of degree~1. 
  In particular, any Hamiltonian path has to have these vertices as endpoints.
  This problem is \NP-hard; see for instance~\cite{GareyJS76,Arkin09}, which in fact proves that \textsc{Hamiltonian \emph{Cycle}} is \NP-complete in subcubic graphs.
  In the latter paper---which shows a~stronger result---the constructed graphs have two adjacent vertices of degree~2.
  Thus the edge between them can be removed to show the \NP-hardness of our Hamiltonian path problem. 
  
  Let $H$ be an $\nu$-vertex subcubic graph, and $s, t \in V(H)$ be its two vertices of degree~1.
  We build an $n$-vertex instance $G$ of \rel$($\textsc{Hamiltonian Cycle}$, \diste, n^{\frac{1}{2}-\beta})$.

  \medskip
  
  \textbf{Construction of~$\bm{G}$.}
  Let $q := \lfloor n^{\frac{1}{2}-\beta} \rfloor + 1$, and make $2q$ copies $H_1, \ldots, H_{2q}$ of $H$.
  For every $i \in [2q]$, denote by $s_i, t_i$ the copies of $s, t$, respectively, in~$H_i$.
  Identify $t_i$ and $s_{i+1}$ for every $i \in [2q-1]$, and $s_1$ and $t_{2q}$.
  Let us call $G'$ the graph obtained at this stage.

  Now turn every vertex $u$ in $V(G') \setminus \bigcup_{i \in [2q]} \{s_i,t_i\}$ (actually $\bigcup_{i \in [2q]} \{s_i,t_i\} = \bigcup_{i \in [2q]} \{s_i\}$) into a~clique $C_u$ of size $q+2$.
  The vertices $s_i$ and $t_i$ remain single vertices.
  For every edge $uv \in E(G')$ not incident to an~$s_i$ or $t_i$, add a~new vertex $x_{uv}$ ($=x_{vu}$).
  Make $x_{uv}$ fully adjacent to $C_u$ and to~$C_v$ (there is no edge between $C_u$ and $C_v$).
  For every $i \in [2q]$, let $s'_i$ (resp.~$t'_i$) be the unique neighbor of $s_i$ (resp.~of~$t_i$) in $H_i$.
  Make $s_i$ fully adjacent to~$C_{s'_i}$, and $t_i$ fully adjacent to~$C_{t'_i}$.
  For convenience, we may use $x_{s_is'_i}, x_{t_it'_i}$ as aliases for $s_i, t_i$, respectively. 
  This finishes the construction of~$G$.

  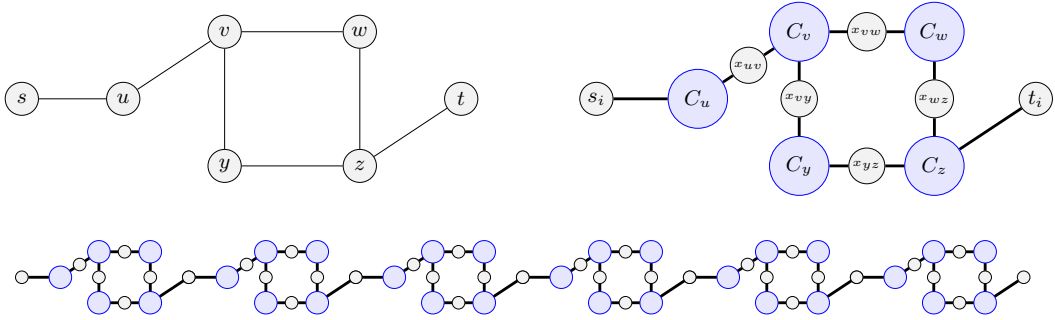
\begin{figure}[h!]
    \centering
\resizebox{\textwidth}{!}{
\begin{tikzpicture}[
   vertex/.style={circle, draw=black, fill=gray!10, minimum size=14pt, inner sep=1pt},
   clique/.style={circle,draw=blue,fill=blue!10,minimum size=25pt,inner sep=1pt},
   cliqueinv/.style={circle,minimum size=25pt,inner sep=1pt},
   font=\small]

  \def\hs{8.5}
  
  \foreach \label/\x/\y in {s/0/0, u/1.5/0, v/3/1, w/5/1, y/3/-1, z/5/-1, t/6.5/0} {
    \node[vertex] (\label) at (\x, \y) {$\label$};
  }
  
  \foreach \label/\x/\y in {u/1.5/0, v/3/1, w/5/1, y/3/-1, z/5/-1} {
    \node[clique] (C\label) at (\x + \hs, \y) {$C_{\label}$};
  }
  \foreach \label/\x/\y in {s/0/0, t/6.5/0} {
    \node[vertex] (C\label) at (\x + \hs, \y) {$\label_i$};
  }

  \foreach \a/\b in {s/u,u/v,v/w,w/z,z/t,v/y,y/z}{
    \draw (\a) -- (\b) ;
  }

\foreach \a/\b in {u/v, v/w, w/z, v/y, y/z}{
  \path let \p1=(C\a),\p2=(C\b) in coordinate (mid-\a-\b) at ($(\p1)!.5!(\p2)$);
  \node[vertex] (x-\a-\b) at ($(mid-\a-\b)$) {\tiny{$x_{\a\b}$}};
}
\foreach \a/\b in {Cs/Cu, x-u-v/Cu, x-u-v/Cv, x-v-y/Cv, x-v-y/Cy, x-v-w/Cv, x-v-w/Cw, x-y-z/Cy, x-y-z/Cz, x-w-z/Cw, x-w-z/Cz, Cz/Ct}{
    \draw[very thick] (\a) -- (\b) ;
}

\def\q{6}
\begin{scope}[scale=0.38, transform shape, xshift=-15cm, yshift=-7cm]
  \foreach \i in {1,...,\q}{
    \begin{scope}[xshift=6.5 * \i cm]
  \foreach \label/\x/\y in {u/1.5/0, v/3/1, w/5/1, y/3/-1, z/5/-1} {
    \node[clique] (C\label) at (\x + \hs, \y) {};
  }
  \foreach \label/\x/\y in {s/0/0, t/6.5/0} {
    \node[vertex] (C\label) at (\x + \hs, \y) {};
  }

\foreach \a/\b in {u/v, v/w, w/z, v/y, y/z}{
  \path let \p1=(C\a),\p2=(C\b) in coordinate (mid-\a-\b) at ($(\p1)!.5!(\p2)$);
  \node[vertex] (x-\a-\b) at ($(mid-\a-\b)$) {};
}
\foreach \a/\b in {Cs/Cu, x-u-v/Cu, x-u-v/Cv, x-v-y/Cv, x-v-y/Cy, x-v-w/Cv, x-v-w/Cw, x-y-z/Cy, x-y-z/Cz, x-w-z/Cw, x-w-z/Cz, Cz/Ct}{
    \draw[very thick] (\a) -- (\b) ;
}
\end{scope}
}
\end{scope}

\end{tikzpicture}
}
\caption{Top left: A subcubic graph $H$ with two degree-1 vertices $s$ and $t$.
  Top right: The graph isomorphic to each $G_i$.
  Larger blue vertices are cliques of size $q+2$, and edges incident to them go toward all their vertices.
Bottom: The obtained graph $G$ with $q=3$, where the leftmost vertex $s_1$ and the rightmost vertex $t_{2q}$ are in fact the same vertex.}
\label{fig:ham-reduction}
\end{figure}
  
  We denote by $G_i$ the induced subgraph of~$G$ derived from~$H_i$.
  Observe that $$|V(G)|=2q ((|V(H)|-2)(q+2) + |E(H)| - 1)=\Theta(q^2 \nu).$$
  Thus $n=\Theta(n^{1-2\beta} \nu)$, which implies $n = \Theta(\nu^{\frac{1}{2 \beta}})$.
  As $\beta>0$ is a~fixed constant, this defines a~polynomial-time reduction.
  See~\cref{fig:ham-reduction} for illustrations.

  \medskip

  \textbf{Correctness.}
  The reduction relies on the following two claims.
  \begin{claim}\label{clm:ham}
    If $H$ admits an $s$--$t$ Hamiltonian path, then deleting (equivalently, editing) fewer than $q$ edges in $G$ cannot make it non-Hamiltonian.
  \end{claim}
  \begin{proofofclaim}
    Let $G^-$ be any graph obtained from $G$ by deleting at~most~$q-1$ edges.
    We show the following two properties that for every $x_{uv} \neq x_{vw} \in V(G)$
    \begin{compactitem}
    \item there is a~path in $G^-$, starting at~$x_{uv}$, ending at~$x_{vw}$, and whose vertex set is $C_v \cup \{x_{uv},x_{vw}\}$, and
    \item if $v$ has a~third neighbor $y$ (in its copy~$H_i$), then there is a~path in $G^-$, starting at~$x_{uv}$, ending at~$x_{vw}$, and whose vertex set is $C_v \cup \{x_{uv},x_{vw},x_{vy}\}$. 
    \end{compactitem}
    We can then conclude by mimicking an~$s$--$t$ Hamiltonian path $P$ of $H$ in every~$G_i$.
    More precisely, we start at $s_i = x_{s_is'_i}$ (for $i$ going from 1 to $2q$), and obey the following rules, where $P_i$ is the path $P$ in copy~$H_i$.
    When in $x_{uv}$, if the successor of $v$ along $P_i$ is $w$, we continue in $G_i$ with the $x_{uv}$--$x_{vw}$ path of the second item if $v$ has a~third neighbor~$y$ in $H_i$ and $x_{vy}$ was not traversed yet.
    We instead go with the $x_{uv}$--$x_{vw}$ path of the first item, otherwise.

    \medskip

    \textbf{First item.}
    Recall that Ore's theorem~\cite{Ore60} is that every $n$-vertex simple graph such that the sum of the degrees of any two non-adjacent vertices is at~least~$n$ is Hamiltonian.
   This implies that $G^-[C_v]$ is Hamiltonian.
   Indeed, a~non-adjacent pair of vertices in $G^-[C_v]$ has combined degree at~least $2q-(q-2)=q+2$.
   We fix a~Hamiltonian cycle $C$ of $G^-[C_v]$.

   The bipartite graph $G^-[\{x_{uv},x_{vw}\},C_v]$ has at~least $2(q+2)-(q-1)=q+5$ edges.
   This ensures that there are two consecutive vertices $u', w'$ along $C$ such that $x_{uv}u', x_{vw}w'$ are edges of $G^-[C_v \cup \{x_{uv},x_{vw}\}]$.
   If not, $G^-[\{x_{uv},x_{vw}\},C_v]$ would have at~most $q+2$ edges.
   The edge $x_{uv}u'$, the $(q+2)$-vertex path along $C$ from $u'$ to $w'$, and the edge $x_{vw}w'$ define our desired spanning path.

   \medskip

   \textbf{Second item.}
   Observe that $C_v \cup \{x_{vy}\}$ is a~clique on $q+3$ vertices.
   Again by Ore's theorem there is a~Hamiltonian cycle $C$ in $G^-[C_v \cup \{x_{vy}\}]$.
   In $G$, vertices $x_{uv}, x_{vw}$ both have $q+2$ neighbors in $C$.
   Let $s \geqslant 3$ be the number of neighbors of $x_{uv}$ in $C$ remaining in~$G^-$; hence $q+2-s$ edges between $x_{uv}$ and $C_v$ are removed in~$G^-$.
   These $s$ neighbors have themselves a~combined neighborhood~$N$ in~$C - \{x_{vy}\}$ of size at~least $s-1$.
   As at~most $s-3$ edges incident to $x_{vw}$ can be deleted, at~least one vertex of~$N$ is still adjacent to~$x_{vw}$.
   We make the $x_{uv}$--$x_{vw}$ Hamiltonian path in $G^-[C_v \cup \{x_{uv},x_{vw},x_{vy}\}]$ as in the previous paragraph.  
  \end{proofofclaim}

  \begin{claim}\label{clm:no-ham}
    If $H$ has no $s$--$t$ Hamiltonian path, then adding (equivalently, editing) fewer than $q$~edges in $G$ cannot make it Hamiltonian.
  \end{claim}
  \begin{proofofclaim}
    Let $G^+$ be any graph obtained from $G$ by adding at~most~$q-1$ edges.
    By the pigeonhole principle, there is at~least one $G_i$ such that no vertex of $G_i - \{s_i,t_i\}$ is incident to an edge added in~$G^+$.
    Thus the set $\{s_i,t_i\}$ separates $V(G_i) \setminus \{s_i,t_i\}$ from the rest of~$G^+$.
    So, for $G^+$ to be Hamiltonian, it should be that $G_i$ admits an $s_i$--$t_i$ Hamiltonian path.
    We shall then just prove that if $G_i$ admits an $s_i$--$t_i$ Hamiltonian path, then so does $H_i$ (hence $H$ has an $s$--$t$ Hamiltonian path).

    Let $u_1, u_2, \ldots, u_{\nu-2}$ be the vertices of $H_i - \{s_i,t_i\}$ in the order in which the cliques $C_u$ are visited for the first time, in some fixed $s_i$--$t_i$ Hamiltonian path $P$ of~$G_i$.
    For each clique~$C_u$, there is a~vertex subset $X_u$ of size at~most~3 (the up to 3 vertices of the form $x_{u,\bullet}$) that disconnects $C_u$ from the rest of the graph in~$G_i$.
    This implies that $C_u \cup X_u$ cannot be exited and reentered by~$P$.
    In particular, this means that $u_1=s'_i$, $u_{\nu-2}=t'_i$, and for every $j \in [\nu-3]$, $u_ju_{j+1} \in E(H_i)$.
    Therefore~$s_i, u_1, u_2, \ldots, u_{\nu-2}, t_i$ is an $s_i$--$t_i$ Hamiltonian path in~$H_i$.
  \end{proofofclaim}
 
  Thus, by~\cref{clm:ham,clm:no-ham}, if a~positive (resp.~negative) instance is reported by the algorithm solving \rel$\-($\textsc{Hamiltonian Cycle}$, \diste, n^{\frac{1}{2}-\beta})$, then $H$ has an $s$--$t$ Hamiltonian path (resp. has no~$s$--$t$ Hamiltonian path).
\end{proof}

We observe that pushing the tractability radius for $\diste$ (from $n/3$) down to $n^{0.99}$ requires improving the current polynomial-time approximation factor for \textsc{Longest Path} in Hamiltonian graphs.
The current best polynomial-time algorithm returns paths of length $\Omega(\frac{\log^2 n}{\log^2 \log n})$ in $n$-vertex Hamiltonian graphs~\cite{Vishwanathan04}.
While better guarantees can be obtained in bounded-degree Hamiltonian graphs~\cite{Feder00}, for any $\beta > 0$, no polynomial-time algorithm is known to output paths of length $n^\beta$ even in subcubic Hamiltonian graphs.

\begin{theorem}\label{thm:barrier-longest-path}
  For any $\beta > 0$, a~polynomial-time algorithm for \rel$($\textsc{Hamiltonian Cycle}$, \diste, n^{1-\beta})$ implies a~polynomial-time detection of paths of length $\lfloor n^\beta \rfloor-1$ in Hamiltonian graphs.
\end{theorem}
\begin{proof}
  We make the same reduction as in~\cref{thm:ham-cycle-hard} with only two copies. 
  Let $H, s, t$ be as in the previous proof, and $\nu$ be the number of vertices of~$H$.
  Let $q := \lfloor n^{1-\beta} \rfloor + 3$.
  We build an $n$-vertex instance $G$ of \rel$($\textsc{Hamiltonian Cycle}$, \diste, n^{1-\beta})$ as shown in~\cref{fig:ham-reduction} with cliques $C_u$ of size $q$ and solely two copies $H_1, H_2$ (and $G_1, G_2$).
  So $s_1=t_2$ and $s_2=t_1$.
  We observe that $G$ has $2((|V(H)|-2)q + |E(H)| - 1) = O(q \nu)$ vertices.
  This implies that $n = O(\nu^{1/\beta})$, so the reduction is indeed polynomial for any fixed $\beta > 0$.

  If the algorithm $\mathcal A$ for \rel$($\textsc{Hamiltonian Cycle}$, \diste, n^{1-\beta})$ outputs a~negative instance (of the form $(G',\nil)$), then we know from the proof of  \cref{thm:ham-cycle-hard} that $H$ has no $s$--$t$ Hamiltonian path, thus $G$ is not Hamiltonian; which does not happen within Hamiltonian graphs.
  The only remaining case is that $\mathcal A$ outputs a~positive instance~$(G',C)$.
  Among the $n$ edges of~$C$, at most $\lfloor n^{1-\beta} \rfloor$ were added to~$G$.
  This still leaves at~least $\lfloor n^\beta \rfloor - 1$ consecutive edges of~$G$ in~$C$.
  We thus obtain the sought path of length $\lfloor n^\beta \rfloor - 1$.
\end{proof}

We now move to the \textsc{Dominating Set} problem, where given a~graph $G$ and an integer $k \in [0,|V(G)|]$, one is asked for a~dominating set of size~$k$, or to output $\nil$ if none exists.

\begin{theorem}\label{thm:ds-algo-diste}
  \rel$(\textsc{Dominating Set}, \diste, n/e)$ is in \FP (where $e$ is Euler's number).
\end{theorem}
\begin{proof}
  Let $(G,k)$ be any input of \rel$(\textsc{Dominating Set}, \diste, n/e)$ with $n = |V(G)|$.
  We first run the known polynomial-time $(1-\frac{1}{e})$-approximation algorithm (actually the greedy algorithm) for \textsc{Max $k$-Coverage}~\cite{Nemhauser78}, the problem of picking $k$ sets from a~given list so as to maximize the cardinality of their union, on the closed-neighborhood set system of~$G$, i.e., $\{N_G[v]~:$ $v \in V(G)\}$.

  If this outputs $k$ sets whose union has size less than $(1-\frac{1}{e})n$, we can safely output $((G,k), \nil)$, as this gives a~guarantee that no dominating set of size~$k$ exists in~$G$.
  Otherwise, we obtain a~set $S$ of $k$ vertices such that $X := V(G) \setminus N_G[S]$ is of size at~most~$\frac{n}{e}$.
  Let us fix some $v \in S$.
  Build the graph $H$ by adding $|X|$ edges to $G$: the edge $vx$ for every $x \in X$.
  We can output $((H,k),S)$ as $\diste(G,H) = |X| \leqslant n/e$ and $S$ is a~dominating set of~$H$ of size~$k$.
\end{proof}

The \emph{blow-up} operation (i.e., replacing every vertex by a~clique) provides a~simple robust reduction for \textsc{Dominating Set} under $\diste$. 

\begin{theorem}\label{thm:ds-hardness-diste}
  For any $\beta > 0$, \rel$(\textsc{Dominating Set}, \diste, n^{1-\beta})$ is \NP-hard.
\end{theorem}
\begin{proof}
  Let $H$ be an $\nu$-vertex instance of \textsc{Dominating Set}.
  We build the $n$-vertex graph $G$ by replacing each vertex $v \in V(H)$ by a~clique $C_v$ of size $q := 2 \lfloor n^{1-\beta} \rfloor + 1$.
  Thus, two distinct vertices $x \in C_u, y \in C_v$ are adjacent in $G$ whenever $u = v$ or $uv \in E(H)$.
  We have $n = q \nu$, thus $n = O(\nu^{1/\beta})$.
  This reduction is polynomial for any fixed $\beta > 0$.

  We give to $G$ the same threshold as $H$, say $k$.
  We claim that a~polynomial-time algorithm solving \rel$(\textsc{Dominating}$ $\textsc{Set}, \diste, n^{1-\beta})$ on input $(G,k)$ would also solve \textsc{Dominating Set} on input $(H,k)$.
  The correctness is given by the following two claims.
  
  \begin{claim}\label{clm:dom}
    If $H$ admits a~dominating set of size $k$, then every graph $G'$ obtained from $G$ by removing at~most~$\lfloor n^{1-\beta} \rfloor$ edges has a~dominating set of size~$k$.
  \end{claim}
  \begin{proofofclaim}
    Note that for every $v \in V(H)$, there is at~least one vertex of $C_v$, say $x_v$, not incident to any edge of $G \triangle G'$ (removed edge).
    This is because $|C_v| = 2 \lfloor n^{1-\beta} \rfloor + 1$.
    Let $S \subseteq V(H)$ be a~dominating set of~$H$, and $S' := \{x_v~:~v \in S\}$.
    By construction, $|S'| = |S|$ and $x_v$ dominates in $G'$ the set $\bigcup_{w \in N_H[v]} C_w$.
    As $S$ is a~dominating set of~$H$, $S'$ is a~dominating set of~$G'$. 
  \end{proofofclaim}

  \begin{claim}\label{clm:no-dom}
    If $H$ has no dominating set of size $k$, then every graph $G'$ obtained from $G$ by adding at~most~$\lfloor n^{1-\beta} \rfloor$ edges has no dominating set of size~$k$.
  \end{claim}
  \begin{proofofclaim} 
    Let $S \subseteq V(G')$ be any set of size~$k$.
    Consider the set $S' := \{v \in V(H)~:~S \cap C_v \neq \emptyset\}$ of size at~most~$k$.
    By assumption, $S'$ is not a~dominating set of~$H$. 
    So there is a~$v \in V(H)$ such that $N_H[v] \cap S' = \emptyset$.
    By the observation in~\cref{clm:dom}, there is a~vertex $x_v \in C_v$ such that $x_v$ is not incident to any edge of $G \triangle G'$ (added edge).
    Vertex $x_v$ is thus not dominated by~$S$.
  \end{proofofclaim}
  
  We can conclude as a~positive (resp.~negative) answer to \rel$(\textsc{Dominating}$ $\textsc{Set}, \diste,$ $n^{1-\beta})$ on input $(G,k)$ implies a~positive (resp.~negative) answer to \textsc{Dominating Set} on input $(H,k)$.
\end{proof}

\Cref{thm:ds-hr-diste} then follows by~\cref{thm:ds-hardness-diste,thm:ds-algo-diste}.

\section{\textsc{Independent Set}, \textsc{Clique}, \textsc{Vertex Cover}, and \textsc{Coloring}}\label{sec:is}

We first show that \textsc{Independent Set} has hardness radius $n^{\frac{1}{2}-o(1)}$ under $\distd$, and $n^{\frac{4}{3}-o(1)}$ under $\diste$.
We recall that \textsc{Independent Set} inputs a~pair $(G,k)$ where $G$ is a~graph and $k$ is a~non-negative integer, and asks for an independent set of size $k$ in~$G$.
We could treat \textsc{Independent Set} purely as a~decision problem (with output in $\{\true, \false\}$).
However all the results presented in this section extend to the setting where the desired output is an actual independent set of size~$k$ in $G$, and $\nil$ if none exists. 

We start by observing that $\sqrt n$ upper bounds the tractability radius for $\distd$.

\begin{theorem}\label{thm:algo-is}
  \rel$($\textsc{Independent Set}$, \distd, n^{\frac{1}{2}})$ is in~\FP.
\end{theorem}
\begin{proof}
  Let $(G,k)$ be the input, with $n = |V(G)|$.
  If $k \leqslant \sqrt n + 1$, fix any subset $S \subseteq V(G)$ of size $\lfloor \sqrt n \rfloor + 1$.
  Build the graph $H$ with vertex set $V(G)$ such that $H[S]$ is an independent set, and $H - S = G - S$.
  By construction, $\distd(G,H) \leqslant \sqrt n$.
  We can thus output $((H,k),S)$.

  If instead $k > \sqrt n + 1$, arbitrarily partition $V(G)$ into $k-1$ parts $V_1, \ldots, V_{k-1}$ of size $\lfloor \frac{n}{k-1} \rfloor$ or $\lceil \frac{n}{k-1} \rceil$.
  Build the graph $H$, obtained from $G$ by turning each $V_i$ into a~clique.
  As $|V_i|-1 \leqslant \sqrt n$, it holds that $\distd(G,H) \leqslant \sqrt n$.
  The partition into $k-1$ cliques $V_1, \ldots, V_{k-1}$ in $H$ implies that $\alpha(H) \leqslant k-1$.
  We can thus output $((H,k),\nil)$.
\end{proof}

Perhaps surprisingly, the easy scheme of~\cref{thm:algo-is} is essentially best possible.
Using the known strong inapproximability of \textsc{Max Independent Set}, we show that the sidestepping problem becomes \NP-hard as soon as the radius gets only slightly smaller.

\begin{theorem}\label{thm:hardness-is}
  For any $\beta > 0$, \rel$($\textsc{Independent Set}$, \distd, n^{\frac{1}{2}-\beta})$ is \NP-hard.
\end{theorem}
\begin{proof}
  For any $\varepsilon \in (0,1/2]$, given $n$-vertex input graphs $G$ satisfying either $\alpha(G) \leqslant n^\varepsilon$ or $\alpha(G) \geqslant n^{1-\varepsilon}$ it is \NP-hard to tell which of the two outcomes holds~\cite{Hastad96,Zuckerman07}.

  Consider inputs $(G,\lfloor \sqrt n \rfloor)$ of \rel$($\textsc{Independent Set}$, \distd,   n^{\frac{1}{2}-\beta})$ with $n = |V(G)|$ and $G$ satisfies the above promise.
  We show that if these instances could be solved in polynomial time, this would contradict, under \P\,$\neq$\,\NP, the latter hardness-of-approximation result with $\varepsilon := \beta/2$.
  Assume that the supposed algorithm returns $((H, \lfloor \sqrt n \rfloor),S)$ such that $\distd(G,H) \leqslant n^{\frac{1}{2}-\beta}$ and $S$ is an independent set of~$H$ of size $\lfloor \sqrt n \rfloor$.
  An important observation is that $G[S]$ has maximum degree at~most $n^{\frac{1}{2}-\beta}$.
  Hence, in polynomial time, one can compute a~subset $S' \subseteq S$ such that $S'$ is an independent set of~$G$ of size at~least
  \[\frac{\lfloor \sqrt n \rfloor}{n^{\frac{1}{2}-\beta}+1} \geqslant \frac{n^\beta}{2} > n^\varepsilon.\]
  For the last inequality to hold, we assume that $n > 2^{1/\varepsilon}$, which we can safely do, as otherwise the instance can be solved by the brute-force approach.
  Therefore $S'$ witnesses that $\alpha(G) \leqslant n^\varepsilon$ does not hold.
  We can thus report that $\alpha(G) \geqslant n^{1-\varepsilon}$.
  
  Now we assume instead that the algorithm returns $((H, \lfloor \sqrt n \rfloor),\nil)$ such that $\distd(G,H) \leqslant n^{\frac{1}{2}-\beta}$ and $\alpha(H) < \sqrt n$.
  We claim that $\alpha(G) < n^{1-\varepsilon}$.
  Indeed, symmetrically to what has been previously observed, if $G$ had an independent set $S$ of size at~least $n^{1-\frac{\beta}{2}}$, then a~$(n^{\frac{1}{2}-\beta}+1)^{-1}$ fraction of $S$ would form an independent set in~$H$.
  Again, assuming that $n > 2^{1/\varepsilon}$, this would contradict $\alpha(H) < \sqrt n$.
  Hence, we can report that the outcome $\alpha(G) \leqslant n^\varepsilon$ holds.
\end{proof}

We can adapt the proof of~\cref{thm:algo-is} to yield an essentially optimal upper bound for distance~$\diste$.
For that, we simply set the demarcation on the independent-set size at $n^{\frac{2}{3}}$ rather than $\sqrt n$.

\begin{theorem}\label{thm:algo-is-diste}
  \rel$($\textsc{Independent Set}$, \diste, n^{\frac{4}{3}})$ is in~\FP.
\end{theorem}
\begin{proof}
  Let $(G,k)$ be the input, with $n = |V(G)|$.
  If $k \leqslant n^{\frac{2}{3}}$, fix any subset $S \subseteq V(G)$ of size~$n^{\frac{2}{3}}$.
  Remove the at~most ${n^{2/3} \choose 2} \leqslant n^{\frac{4}{3}}$ edges of $G[S]$ from $G$, and let $H$ be the obtained graph.
  We can output $((H,k),S)$ since $\diste(G,H) \leqslant n^{\frac{4}{3}}$ and $S$ is an independent set of size~$k$ in~$H$.

  If instead $k > n^{\frac{2}{3}}$, arbitrarily partition $V(G)$ into $\lfloor n^{\frac{2}{3}} \rfloor$ parts $V_1, \ldots, V_{\lfloor n^{2/3} \rfloor}$ each of size at~most $\lceil n^{\frac{1}{3}} \rceil +1$.
  Turn each $V_i$ into a~clique, and call the obtained graph~$H$.
  This adds to~$G$ at most $\lfloor n^{2/3} \rfloor \cdot {\lceil n^{1/3} \rceil+1 \choose 2} \leqslant n^{\frac{4}{3}}$ edges, for large enough~$n$.
  Thus we can output $((H,k),\nil)$.
\end{proof}

Finally we show the counterpart of~\cref{thm:hardness-is} for $\diste$.

\begin{theorem}\label{thm:hardness-is-diste}
  For any $\beta > 0$, \rel$($\textsc{Independent Set}$, \diste, n^{\frac{4}{3}-\beta})$ is \NP-hard.
\end{theorem}
\begin{proof}
  As in~\cref{thm:hardness-is}, we use the hardness of distinguishing, for any $\varepsilon \in (0,1/2]$, $n$-vertex graphs $G$ such that $\alpha(G) \leqslant n^\varepsilon$ from those such that $\alpha(G) \geqslant n^{1-\varepsilon}$.

  Consider inputs $(G,\lfloor n^{2/3} \rfloor)$ of \rel$($\textsc{Independent Set}$, \diste, n^{\frac{4}{3}-\beta})$ with $n = |V(G)|$ and $G$ satisfies one of the previous outcomes.
  We show that solving these instances in polynomial time would contradict, unless \P\,$=$\,\NP, the above hardness-of-approximation result with $\varepsilon := \beta/4$.
  Assume that the supposed algorithm returns $((H, \lfloor n^{2/3} \rfloor),S)$ such that $\diste(G,H) \leqslant n^{\frac{4}{3}-\beta}$ and $S$ is an independent set of~$H$ of size $\lfloor n^{2/3} \rfloor$.

  We observe that $G[S]$ has at~most $n^{\frac{4}{3}-\beta}$ edges.
  Hence every (induced) subgraph of $G[S]$ has a~vertex of degree at~most~$2 (n^{\frac{4}{3}-\beta})^{1/2}=2 n^{\frac{2}{3}-2 \varepsilon}$.
  Thus $G[S]$ has an independent set of size at~least \[\frac{\lfloor n^{2/3} \rfloor}{2 n^{\frac{2}{3} - 2 \varepsilon} + 1} > n^\varepsilon,\]
  for $n$ at~least a~fixed function of~$\varepsilon$.
  Therefore $\alpha(G) \leqslant n^\varepsilon$ cannot hold, so we can conclude that $\alpha(G) \geqslant n^{1-\varepsilon}$.

  We now assume that the algorithm returns instead $((H, \lfloor n^{2/3} \rfloor),\nil)$ such that $\diste(G,H) \leqslant n^{\frac{4}{3}-\beta}$ and $\alpha(H) < \lfloor n^{2/3} \rfloor$.
  We claim that $\alpha(G) < n^{1-\varepsilon}$.
  Suppose for the sake of contradiction that $G$ has an independent set $S$ of size $n^{1-\frac{\beta}{4}}$.
  Note that $H[S]$ has at~most~$n^{\frac{4}{3}-\beta}$ edges.
  In particular, $H[S]$ has at~most~$n^{1-2 \varepsilon}$ vertices of degree at~least~$2 n^{\frac{1}{3}-2 \varepsilon}$, thus more than $n^{1- \varepsilon} - n^{1-2 \varepsilon}$ vertices of degree less than~$2 n^{\frac{1}{3}-2 \varepsilon}$.
  Therefore $H[S]$ has an independent set of size at~least \[\frac{n^{1-\varepsilon}-n^{1-2\varepsilon}}{2 n^{\frac{1}{3}-2 \varepsilon} + 1} > \frac{1}{10} n^{\frac{2}{3} + \varepsilon} > n^{2/3},\]
  for some large enough~$n$.
  This would contradict that $\alpha(H) < \lfloor n^{2/3} \rfloor$.
  Hence, we can conclude that $\alpha(G) \leqslant n^\varepsilon$ holds.
\end{proof}

We now introduce the necessary formalism to show that all the previous results of this section carry over to \textsc{Clique} and \textsc{Vertex Cover}.

A~\emph{polynomial-time isomorphism} $\varphi$ from $\Pi$ to $\Pi'$ (where $\Pi$ and $\Pi'$ are not necessarily decision problems) is a~bijective map from the inputs of $\Pi$ to the inputs of $\Pi'$ that is a~polynomial-time reduction, with inverse $\varphi^{-1}$ also computable in polynomial time.\footnote{For decision problems, this notion is at the core of the Berman--Hartmanis conjecture stating that there is a~polynomial-time isomorphism between any two \NP-complete languages.}
Specifically when $\Pi$ and $\Pi'$ are function problems, the bijection $\varphi$ is accompanied by a~polynomial-time function $\psi$, called \emph{output pullback} (of $\varphi$), such that for any $x \in \sol_{\Pi'}(\varphi(I))$, it holds that $\psi(x) \in \sol_\Pi(I)$.
Note that if instead $\Pi$ and $\Pi'$ are decision problems, the function $\psi$ would simply map $\true$ to $\true$ and $\false$ to $\false$, and can remain implicit.
We say that $\varphi$ is \emph{length-preserving} if for every input $I$ of $\Pi$, $|\varphi(I)|=|I|$.

\begin{theorem}\label{thm:lpi}
  Let $\Pi$ and $\Pi'$ be two problems, $\varphi$ be a~length-preserving polynomial-time isomorphism from $\Pi$ to~$\Pi'$, $\psi$ be its output pullback, and $\dist$ (resp.~$\dist'$) be a~map whose domain is the set of pairs of inputs of~$\Pi$ (resp.~$\Pi'$) such that for every pair of inputs $I, J$ of~$\Pi$, $\dist(I,J) = \dist'(\varphi(I),\varphi(J))$.

  For any map $d$, if \rel$(\Pi', \dist', d)$ is in \FP, then \rel$(\Pi, \dist, d)$ is in \FP, and if \rel$(\Pi, \dist, d)$ is \NP-hard, then \rel$(\Pi', \dist', d)$ is \NP-hard.
\end{theorem}

\begin{proof}
  We describe a~Turing reduction from \rel$(\Pi, \dist, d)$ to \rel$(\Pi', \dist', d)$.
  
  On input $I$ of $\Pi$, compute the instance $I' := \varphi(I)$ of~$\Pi'$ in polynomial time.
  Run the polynomial-time algorithm for \rel$(\Pi', \dist', d)$ on input~$I'$.
  It outputs $(J',\Pi'(J'))$ with $\dist'(I',J') \leqslant d(|I'|)$.
  Now compute $J := \varphi^{-1}(J')$, and $x := \psi(\Pi'(J'))$.
  Finally output~$(J,x)$.
  By definition of the output pullback, $x$ is indeed a~correct solution to $\Pi$ on~$J$.
  Furthermore, $\dist(I,J) = \dist'(\varphi(I),\varphi(J)) = \dist'(I',J') \leqslant d(|I'|) = d(|\varphi(I)|) = d(|I|)$, since $\varphi$ is length-preserving.
\end{proof}

We now use~\cref{thm:lpi} for \textsc{Clique} and \textsc{Vertex Cover}.
Similarly to \textsc{Independent Set}, \textsc{Clique} (resp.~\textsc{Vertex Cover}) inputs a~pair $(G,k)$ where $G$ is a~graph and $k \in [0,|V(G)|]$, and a~clique of~$G$ (resp.~vertex cover of~$G$) of size $k$ is expected as output, and $\nil$ if none exists. 

\begin{theorem}\label{thm:clique-vc}
  \textsc{Clique} and \textsc{Vertex Cover} also have hardness radius $n^{\frac{1}{2}-o(1)}$ for $\distd$, and $n^{\frac{4}{3}-o(1)}$ for $\diste$.
\end{theorem}

\begin{proof}
  We start with the case of~\textsc{Clique}.
  The map $\varphi: (G,k) \mapsto (\overline G, k)$, where $\overline G$ is the complement of graph~$G$ and $k$ is a~non-negative integer, is a~length-preserving polynomial-time isomorphism from \textsc{Clique} to \textsc{Independent Set} and from \textsc{Independent Set} to \textsc{Clique} (as $\varphi$ is an involution), and output pullback $\psi$ defined as the identity map.
  Indeed $|V(G)|=|V(\overline G)|$, $\overline{\overline G}=G$, and an independent set in $\overline G$ is a~clique in~$G$.
  Moreover $G \triangle H = \overline G \triangle \overline H$, so $\dist((G,k),(H,k')) = \dist(\varphi(G,k),\varphi(H,k'))$ for $\dist \in \{\diste, \distd\}$.

  Thus the assumptions of~\cref{thm:lpi} are met by $\varphi$ and $\psi$ for $$(\Pi,\Pi') \in \{(\textsc{Clique}, \textsc{Independent Set}), (\textsc{Independent Set},\textsc{Clique})\}$$ and $\dist=\dist' \in \{\diste,\distd\}$.
  Hence~\cref{thm:lpi,thm:algo-is,thm:algo-is-diste} imply that \rel$($\textsc{Clique}$,$ $\distd, n^{\frac{1}{2}})$ and \rel$($\textsc{Clique}$, \diste, n^{\frac{4}{3}})$ are in~\FP.
  Now considering the reduction from \textsc{Independent Set} to \textsc{Clique}, \cref{thm:lpi,thm:hardness-is,thm:hardness-is-diste} imply that, for any $\beta > 0$, \rel$($\textsc{Clique}$,$ $\distd, n^{\frac{1}{2}-\beta})$ and \rel$($\textsc{Clique}$, \diste, n^{\frac{4}{3}-\beta})$ are \NP-hard.

  We now deal with~\textsc{Vertex Cover}.
  We define the involutive map $\varphi: (G,k) \mapsto (G, |V(G)|-k)$.
  The map $\varphi$ is a~length-preserving polynomial-time isomorphism from \textsc{Vertex Cover} to \textsc{Independent Set} and from \textsc{Independent Set} to \textsc{Vertex Cover}, with output pullback $\psi$ defined as the involution $X \mapsto V(G) \setminus X$.
  Indeed $X$ is a~vertex cover of $G$ if and only if $V(G) \setminus X$ is an independent set of~$G$.
  We have that $$\dist((G,k),(H,k')) = \dist((G,|V(G)|-k),(H,|V(H)|-k')) = \dist(\varphi(G,k),\varphi(H,k'))$$ for $\dist \in \{\diste, \distd\}$.
  This holds since these distances are infinite if $V(G) \neq V(H)$ or $k \neq k'$, and are equal to $\dist(G,H)$ otherwise.

  Thus \cref{thm:lpi,thm:algo-is,thm:algo-is-diste} imply that \rel$($\textsc{Vertex Cover}$,$ $\distd, n^{\frac{1}{2}})$ and \rel$($\textsc{Vertex Cover}$, \diste, n^{\frac{4}{3}})$ are in~\FP, whereas \cref{thm:lpi,thm:hardness-is,thm:hardness-is-diste} imply that, for any $\beta > 0$, \rel$($\textsc{Vertex Cover}$,$ $\distd, n^{\frac{1}{2}-\beta})$ and \rel$($\textsc{Vertex Cover}$, \diste, n^{\frac{4}{3}-\beta})$ are \NP-hard.
\end{proof}

We finish this section with the case of \textsc{Coloring} and \textsc{Clique Cover}.

We first observe that \cref{thm:algo-is,thm:algo-is-diste} also hold for \textsc{Coloring}, which takes as input a~graph $G$ and an integer~$k \in [0,|V(G)|]$, and asks for a~partition of $V(G)$ into $k$ independent sets if one exists, and to output $\nil$ otherwise.
To do so, one can adapt these theorems for the \textsc{Clique Cover} problem, which boils down to \textsc{Coloring} in the complement graph, and invoke~\cref{thm:lpi}.
Indeed, in the algorithms of \cref{thm:algo-is,thm:algo-is-diste}, the used witnesses of small independence number (negative instances for \textsc{Independent Set}) are vertex-partitions into cliques (positive instances for \textsc{Clique Cover}), whereas a~large independent set (positive instances for \textsc{Independent Set}) is always a~witness that the graph needs a~large number of cliques to be vertex-partitioned (negative instances for \textsc{Clique Cover}).

\begin{theorem}\label{thm:algo-coloring}
  For any $\Pi \in \{\textsc{Coloring}, \textsc{Clique Cover}\}$, \rel$(\Pi, \distd, n^{\frac{1}{2}})$ and \rel$(\Pi, \diste, n^{\frac{4}{3}})$ are in~\FP.
\end{theorem}

\cref{thm:hardness-is,thm:hardness-is-diste} also hold for \textsc{Coloring} (and \textsc{Clique Cover}), but their proofs require a~bit more adjustments.

\begin{theorem}\label{thm:hardness-coloring}
For any $\beta > 0$, \rel$($\textsc{Coloring}$, \distd, n^{\frac{1}{2}-\beta})$ is \NP-hard.
\end{theorem}
\begin{proof}
  For any $\varepsilon \in (0,1/2]$, given $n$-vertex input graphs $G$ satisfying either $\chi(G) \leqslant n^\varepsilon$ or $\chi(G) \geqslant n^{1-\varepsilon}$, it is \NP-hard to tell which of the two outcomes holds~\cite{Hastad96,Zuckerman07}.
 Consider inputs $(G,\lfloor \sqrt n \rfloor)$ of \rel$($\textsc{Coloring}$, \distd, n^{\frac{1}{2}-\beta})$ with $n = |V(G)|$ and $G$ satisfies that $\chi(G) \leqslant n^\varepsilon$ or $\chi(G) \geqslant n^{1-\varepsilon}$.
  We show that a~polynomial-time algorithm for these instances would contradict, unless \P\,$=$\,\NP, the above hardness-of-approximation result with $\varepsilon := \beta/2$.

  Assume that the supposed algorithm returns $((H, \lfloor \sqrt n \rfloor),\mathcal P)$ such that $\distd(G,H) \leqslant n^{\frac{1}{2}-\beta}$ and $\mathcal P$ is a~partition of~$V(H)$ into $\lfloor \sqrt n \rfloor$ independent sets of~$H$.
  As $G' := G-E(H)$ (the graph $G$ deprived of the edges of~$H$) has maximum degree at~most $n^{\frac{1}{2}-\beta}$, there is a~partition $\mathcal Q$ of $V(G') = V(H)$ into $n^{\frac{1}{2}-\beta}+1$ independent sets of~$G'$.
  Partition $\mathcal Q$ refines $\mathcal P$ (split each part of~$\mathcal P$ into its intersections with parts of $\mathcal Q$) into a~partition $\mathcal P'$ of $V(G)=V(H)$ with at~most~$\lfloor \sqrt n \rfloor \cdot (n^{\frac{1}{2}-\beta}+1)$ parts.
  By construction, each part of $\mathcal P'$ is an independent set of~$G$.
  Moreover $\lfloor \sqrt n \rfloor \cdot (n^{\frac{1}{2}-\beta}+1) \leqslant n^{1-\beta} + \sqrt n < n^{1-\varepsilon}$ for large enough~$n$.
  Thus we can report that $\chi(G) \leqslant n^\varepsilon$.
  
  Now we assume instead that the algorithm returns $((H, \lfloor \sqrt n \rfloor),\nil)$ such that $\distd(G,H) \leqslant n^{\frac{1}{2}-\beta}$ and $\chi(H) > \sqrt n$.
  We claim that $\chi(G) > n^\varepsilon$.
  Assume, for the sake of contradiction, that there is a~partition $\mathcal P$ of~$V(G) = V(H)$ into $\lfloor n^\varepsilon \rfloor$ independent sets of~$G$.
  As $H' := H - E(G)$ has maximum degree at~most $n^{\frac{1}{2}-\beta}$, there is a~partition $\mathcal Q$ of $V(H') = V(H)$ into $n^{\frac{1}{2}-\beta}+1$ independent sets of~$H'$.
  Then the refinement $\mathcal P'$ of~$\mathcal P$ by $\mathcal Q$ is a~partition of $V(H)$ into at~most $\lfloor n^\varepsilon \rfloor \cdot (n^{\frac{1}{2}-\beta}+1) \leqslant n^{\frac{1}{2}-\beta+\varepsilon}+ n^\varepsilon < \sqrt n$ independent sets of~$H$; a~contradiction to $\chi(H) > \sqrt n$.
  Hence, we can report that $\chi(G) \geqslant n^{1-\varepsilon}$ holds.
\end{proof}

We finally adapt~\cref{thm:hardness-is-diste} for \textsc{Coloring}.
We need the following folklore observation.

\begin{observation}\label{obs:m-to-degen}
  Every $m$-edge graph has chromatic number at~most~$\lceil \sqrt{2m} \rceil$.
\end{observation}
\begin{proof}
  It is in fact enough to show that any $m$-edge graph $G$ has a~vertex of degree at~most $\lceil \sqrt{2m} \rceil - 1$.
  If all the vertices of~$G$ have degree at~least $\lceil \sqrt{2m} \rceil$, then $m \geqslant \frac{1}{2} n \lceil \sqrt{2m} \rceil$.
  But also $n \geqslant \lceil \sqrt{2m} \rceil+1$ for a~vertex to possibly have $\lceil \sqrt{2m} \rceil$ neighbors.
  So $m \geqslant \frac{1}{2} \sqrt{2m} (\sqrt{2m}+1)> m$; a~contradiction.
\end{proof}

The main additional difference compared to the adaptation of~\cref{thm:hardness-coloring} is the use of the Cauchy--Schwarz inequality to get better upper bounds on some chromatic numbers.

\begin{theorem}\label{thm:hardness-coloring-diste}
  For any $\beta > 0$, \rel$($\textsc{Coloring}$, \diste, n^{\frac{4}{3}-\beta})$ is \NP-hard.
\end{theorem}
\begin{proof}
  As in~\cref{thm:hardness-coloring}, we rely on the hardness of distinguishing, for any $\varepsilon \in (0,1/2]$, $n$-vertex graphs $G$ such that $\chi(G) \leqslant n^\varepsilon$ from those such that $\chi(G) \geqslant n^{1-\varepsilon}$.
  Consider inputs $(G,\lfloor n^{2/3} \rfloor)$ of \rel$($\textsc{Coloring}$, \diste, n^{\frac{4}{3}-\beta})$ with $n = |V(G)|$ and $G$ satisfies $\chi(G) \leqslant n^\varepsilon$ or $\chi(G) \geqslant n^{1-\varepsilon}$.
  We show that solving these instances in polynomial time would contradict, unless \P\,$=$\,\NP, the above hardness-of-approximation result with $\varepsilon := \beta/4$.

  Assume that the supposed algorithm returns $((H, \lfloor n^{2/3} \rfloor),\mathcal P)$ such that $\diste(G,H) \leqslant n^{\frac{4}{3}-\beta}$ and $\mathcal P := \{P_1, \ldots, P_t\}$ is a~partition of~$V(H)$ into $t := \lfloor n^{2/3} \rfloor$ independent sets of~$H$.
  Here we need to be more delicate than in the proof of~\cref{thm:hardness-coloring} in upper bounding the chromatic number of~$G$.
  We will rely on~\cref{obs:m-to-degen} and the Cauchy--Schwarz inequality.

  For every $i \in [t]$, let $a_i$ be the number of edges in~$G[P_i]$.
  As $\diste(G,H) \leqslant n^{\frac{4}{3}-\beta}$, \[\sum\limits_{i \in [t]} a_i \leqslant n^{\frac{4}{3}-\beta}.\]
  By~\cref{obs:m-to-degen}, $G[P_i]$ can be properly colored with at~most $\lceil \sqrt{2a_i} \rceil$ colors.
  Hence the chromatic number of~$G$ is at~most
  $$\sum\limits_{i \in [t]} \lceil \sqrt{2a_i} \rceil \leqslant \sum\limits_{i \in [t]} (\sqrt{2a_i}+1) \leqslant t + \sqrt 2 \sum\limits_{i \in [t]} \sqrt{a_i}.$$
  In $\mathbb R^t$, the Cauchy--Schwarz inequality is that
  \[ \left( \sum\limits_{i \in [t]} x_iy_i \right)^2 \leqslant \left( \sum\limits_{i \in [t]} x_i^2 \right) \left( \sum\limits_{i \in [t]} y_i^2 \right).\]
  Setting $x_i := \sqrt{a_i}$ and $y_i := 1$, we get that
  $$\chi(G) \leqslant t + \sqrt 2 \sqrt{\left( \sum\limits_{i \in [t]} \sqrt{a_i}^2 \right) \left( \sum\limits_{i \in [t]} 1^2 \right)}
  =  t + \sqrt {2 t \left( \sum\limits_{i \in [t]} a_i \right)} \leqslant t + \sqrt{2 t \cdot  n^{\frac{4}{3}-\beta}}.$$
  Thus $$\chi(G) \leqslant \lfloor n^{2/3} \rfloor + \sqrt{2} \cdot n^{1/3} \cdot n^{\frac{2}{3}-2\varepsilon} \leqslant \lfloor n^{2/3} \rfloor + \sqrt{2} n^{1-2\varepsilon} < n^{1-\varepsilon},$$
  for large enough~$n$.
  And we can conclude that $\chi(G) \leqslant n^\varepsilon$.

  We now assume that the algorithm returns instead $((H, \lfloor n^{2/3} \rfloor),\nil)$ such that $\diste(G,H) \leqslant n^{\frac{4}{3}-\beta}$ and $\chi(H) > n^{2/3}$.
  We claim that $\chi(G) > n^\varepsilon$.
  Suppose for the sake of contradiction that there is a~partition $\mathcal Q := \{Q_1, \ldots, Q_s\}$ of $V(G)$ into $s := \lfloor n^\varepsilon \rfloor$ independent sets of~$G$.
  For every $i \in [s]$, let $b_i$ be the number of edges in $H[Q_i]$.
  As $\diste(G,H) \leqslant n^{\frac{4}{3}-\beta}$, $\sum_{i \in [s]} b_i \leqslant n^{\frac{4}{3}-\beta}$.
  By~\cref{obs:m-to-degen}, we have $\chi(H) \leqslant \sum_{i \in [s]} \lceil \sqrt{2 b_i} \rceil$.
  A~similar application of the Cauchy--Schwarz inequality with $x_i := \sqrt{b_i}$ and $y_i := 1$ yields
  $$ \chi(H) \leqslant s + \sqrt{2 s \cdot  n^{\frac{4}{3}-\beta}} \leqslant \lfloor n^\varepsilon \rfloor + \sqrt{2} \cdot n^{\varepsilon/2} \cdot n^{\frac{2}{3}-2\varepsilon} = \lfloor n^\varepsilon \rfloor + \sqrt{2} \cdot n^{\frac{2}{3}-\frac{3 \varepsilon}{2}} < n^{2/3},$$
  for large enough~$n$; a~contradiction to~$\chi(H) > n^{2/3}$.
  Hence, we can conclude that $\chi(G) \geqslant n^{1-\varepsilon}$ holds.
\end{proof}

One can extend \cref{thm:hardness-coloring,thm:hardness-coloring-diste} to \textsc{Clique Cover} by~\cref{thm:lpi}.
This finishes the proof of~\cref{thm:1/2-4/3}.

\paragraph*{Acknowledgment.}
The author thanks Bart M. P. Jansen for telling him about certified algorithms.


\begin{thebibliography}{10}

\bibitem{Arkin09}
Esther~M. Arkin, S{\'{a}}ndor~P. Fekete, Kamrul Islam, Henk Meijer, Joseph
  S.~B. Mitchell, Yurai~N{\'{u}}{\~{n}}ez Rodr{\'{\i}}guez, Valentin
  Polishchuk, David Rappaport, and Henry Xiao.
\newblock Not being (super)thin or solid is hard: {A} study of grid
  hamiltonicity.
\newblock {\em Comput. Geom.}, 42(6-7):582--605, 2009.
\newblock URL: \url{https://doi.org/10.1016/j.comgeo.2008.11.004}, \href
  {https://doi.org/10.1016/J.COMGEO.2008.11.004}
  {\path{doi:10.1016/J.COMGEO.2008.11.004}}.

\bibitem{Arthur11}
David Arthur, Bodo Manthey, and Heiko R{\"{o}}glin.
\newblock Smoothed analysis of the k-{M}eans method.
\newblock {\em J. {ACM}}, 58(5):19:1--19:31, 2011.
\newblock \href {https://doi.org/10.1145/2027216.2027217}
  {\path{doi:10.1145/2027216.2027217}}.

\bibitem{Awasthi12}
Pranjal Awasthi, Avrim Blum, and Or~Sheffet.
\newblock Center-based clustering under perturbation stability.
\newblock {\em Inf. Process. Lett.}, 112(1-2):49--54, 2012.
\newblock URL: \url{https://doi.org/10.1016/j.ipl.2011.10.006}, \href
  {https://doi.org/10.1016/J.IPL.2011.10.006}
  {\path{doi:10.1016/J.IPL.2011.10.006}}.

\bibitem{Banderier03}
Cyril Banderier, Ren{\'{e}} Beier, and Kurt Mehlhorn.
\newblock Smoothed analysis of three combinatorial problems.
\newblock In Branislav Rovan and Peter Vojt{\'{a}}s, editors, {\em Mathematical
  Foundations of Computer Science 2003, 28th International Symposium, {MFCS}
  2003, Bratislava, Slovakia, August 25-29, 2003, Proceedings}, volume 2747 of
  {\em Lecture Notes in Computer Science}, pages 198--207. Springer, 2003.
\newblock \href {https://doi.org/10.1007/978-3-540-45138-9\_14}
  {\path{doi:10.1007/978-3-540-45138-9\_14}}.

\bibitem{barrett2020eco}
Thomas~D. Barrett, William~R. Clements, Jakob~N. Foerster, and Alex~I. Lvovsky.
\newblock Exploratory combinatorial optimization with reinforcement learning.
\newblock In {\em Proceedings of the AAAI Conference on Artificial
  Intelligence}, volume~34, pages 3243--3250, 2020.
\newblock URL: \url{https://ojs.aaai.org/index.php/AAAI/article/view/5723},
  \href {https://doi.org/10.1609/aaai.v34i04.5723}
  {\path{doi:10.1609/aaai.v34i04.5723}}.

\bibitem{Beier06}
Ren{\'{e}} Beier and Berthold V{\"{o}}cking.
\newblock Typical properties of winners and losers in discrete optimization.
\newblock {\em {SIAM} J. Comput.}, 35(4):855--881, 2006.
\newblock \href {https://doi.org/10.1137/S0097539705447268}
  {\path{doi:10.1137/S0097539705447268}}.

\bibitem{Bilu12}
Yonatan Bilu and Nathan Linial.
\newblock Are stable instances easy?
\newblock {\em Comb. Probab. Comput.}, 21(5):643--660, 2012.
\newblock \href {https://doi.org/10.1017/S0963548312000193}
  {\path{doi:10.1017/S0963548312000193}}.

\bibitem{BlaserM15}
Markus Bl{\"{a}}ser and Bodo Manthey.
\newblock Smoothed complexity theory.
\newblock {\em {ACM} Trans. Comput. Theory}, 7(2):6:1--6:21, 2015.
\newblock \href {https://doi.org/10.1145/2656210} {\path{doi:10.1145/2656210}}.

\bibitem{cappart2023gnnco}
Quentin Cappart, Didier Ch{\'e}telat, Elias~B. Khalil, Andrea Lodi, Christopher
  Morris, and Petar Veli{\v{c}}kovi{\'c}.
\newblock Combinatorial optimization and reasoning with graph neural networks.
\newblock {\em Journal of Machine Learning Research}, 24(130):1--61, 2023.
\newblock URL: \url{https://jmlr.org/papers/v24/21-0449.html}.

\bibitem{Crespelle23}
Christophe Crespelle, P{\aa}l~Gr{\o}n{\aa}s Drange, Fedor~V. Fomin, and Petr~A.
  Golovach.
\newblock A survey of parameterized algorithms and the complexity of edge
  modification.
\newblock {\em Comput. Sci. Rev.}, 48:100556, 2023.
\newblock URL: \url{https://doi.org/10.1016/j.cosrev.2023.100556}, \href
  {https://doi.org/10.1016/J.COSREV.2023.100556}
  {\path{doi:10.1016/J.COSREV.2023.100556}}.

\bibitem{Dirac52}
Gabriel~Andrew Dirac.
\newblock Some theorems on abstract graphs.
\newblock {\em Proceedings of the London Mathematical Society}, 3(1):69--81,
  1952.

\bibitem{Englert16}
Matthias Englert, Heiko R{\"{o}}glin, and Berthold V{\"{o}}cking.
\newblock Smoothed analysis of the 2-opt algorithm for the general {TSP}.
\newblock {\em {ACM} Trans. Algorithms}, 13(1):10:1--10:15, 2016.
\newblock \href {https://doi.org/10.1145/2972953} {\path{doi:10.1145/2972953}}.

\bibitem{Feder00}
Tom{\'{a}}s Feder, Rajeev Motwani, and Carlos~S. Subi.
\newblock Finding long paths and cycles in sparse hamiltonian graphs.
\newblock In F.~Frances Yao and Eugene~M. Luks, editors, {\em Proceedings of
  the Thirty-Second Annual {ACM} Symposium on Theory of Computing, May 21-23,
  2000, Portland, OR, {USA}}, pages 524--529. {ACM}, 2000.
\newblock \href {https://doi.org/10.1145/335305.335368}
  {\path{doi:10.1145/335305.335368}}.

\bibitem{GareyJS76}
M.~R. Garey, David~S. Johnson, and Larry~J. Stockmeyer.
\newblock Some simplified {NP}-complete graph problems.
\newblock {\em Theor. Comput. Sci.}, 1(3):237--267, 1976.
\newblock \href {https://doi.org/10.1016/0304-3975(76)90059-1}
  {\path{doi:10.1016/0304-3975(76)90059-1}}.

\bibitem{Hastad96}
Johan H{\aa}stad.
\newblock Clique is hard to approximate within $n^{1-\varepsilon}$.
\newblock In {\em 37th Annual Symposium on Foundations of Computer Science,
  {FOCS} '96, Burlington, Vermont, USA, 14-16 October, 1996}, pages 627--636.
  {IEEE} Computer Society, 1996.
\newblock \href {https://doi.org/10.1109/SFCS.1996.548522}
  {\path{doi:10.1109/SFCS.1996.548522}}.

\bibitem{Jerrum92}
Mark Jerrum.
\newblock Large cliques elude the metropolis process.
\newblock {\em Random Struct. Algorithms}, 3(4):347--360, 1992.
\newblock URL: \url{https://doi.org/10.1002/rsa.3240030402}, \href
  {https://doi.org/10.1002/RSA.3240030402} {\path{doi:10.1002/RSA.3240030402}}.

\bibitem{Lelarge25}
Marc Lelarge.
\newblock Bootstrap learning for combinatorial graph alignment with sequential
  gnns.
\newblock {\em CoRR}, abs/2510.03086, 2025.
\newblock URL: \url{https://doi.org/10.48550/arXiv.2510.03086}, \href
  {http://arxiv.org/abs/2510.03086} {\path{arXiv:2510.03086}}, \href
  {https://doi.org/10.48550/ARXIV.2510.03086}
  {\path{doi:10.48550/ARXIV.2510.03086}}.

\bibitem{Lemos19}
Henrique Lemos, Marcelo O.~R. Prates, Pedro H.~C. Avelar, and Lu{\'{\i}}s~C.
  Lamb.
\newblock Graph colouring meets deep learning: Effective graph neural network
  models for combinatorial problems.
\newblock In {\em 31st {IEEE} International Conference on Tools with Artificial
  Intelligence, {ICTAI} 2019, Portland, OR, USA, November 4-6, 2019}, pages
  879--885. {IEEE}, 2019.
\newblock \href {https://doi.org/10.1109/ICTAI.2019.00125}
  {\path{doi:10.1109/ICTAI.2019.00125}}.

\bibitem{Makarychev20}
Konstantin Makarychev and Yury Makarychev.
\newblock Certified algorithms: Worst-case analysis and beyond.
\newblock In Thomas Vidick, editor, {\em 11th Innovations in Theoretical
  Computer Science Conference, {ITCS} 2020, January 12-14, 2020, Seattle,
  Washington, {USA}}, volume 151 of {\em LIPIcs}, pages 49:1--49:14. Schloss
  Dagstuhl - Leibniz-Zentrum f{\"{u}}r Informatik, 2020.
\newblock URL: \url{https://doi.org/10.4230/LIPIcs.ITCS.2020.49}, \href
  {https://doi.org/10.4230/LIPICS.ITCS.2020.49}
  {\path{doi:10.4230/LIPICS.ITCS.2020.49}}.

\bibitem{Nemhauser78}
George~L. Nemhauser, Laurence~A. Wolsey, and Marshall~L. Fisher.
\newblock An analysis of approximations for maximizing submodular set functions
  - {I}.
\newblock {\em Math. Program.}, 14(1):265--294, 1978.
\newblock \href {https://doi.org/10.1007/BF01588971}
  {\path{doi:10.1007/BF01588971}}.

\bibitem{parkar2025localupdates}
Devendra Parkar, Anya Chaturvedi, and Joshua~J. Daymude.
\newblock Unsupervised learning of local updates for maximum independent set in
  dynamic graphs, 2025.
\newblock To appear at IJCNN 2026.
\newblock URL: \url{https://arxiv.org/abs/2505.13754}, \href
  {http://arxiv.org/abs/2505.13754} {\path{arXiv:2505.13754}}, \href
  {https://doi.org/10.48550/arXiv.2505.13754}
  {\path{doi:10.48550/arXiv.2505.13754}}.

\bibitem{SpielmanT04}
Daniel~A. Spielman and Shang{-}Hua Teng.
\newblock Smoothed analysis of algorithms: Why the simplex algorithm usually
  takes polynomial time.
\newblock {\em J. {ACM}}, 51(3):385--463, 2004.
\newblock \href {https://doi.org/10.1145/990308.990310}
  {\path{doi:10.1145/990308.990310}}.

\bibitem{SpielmanT09}
Daniel~A. Spielman and Shang{-}Hua Teng.
\newblock Smoothed analysis: an attempt to explain the behavior of algorithms
  in practice.
\newblock {\em Commun. {ACM}}, 52(10):76--84, 2009.
\newblock \href {https://doi.org/10.1145/1562764.1562785}
  {\path{doi:10.1145/1562764.1562785}}.

\bibitem{sun2022annealed}
Haoran Sun, Etash~K. Guha, and Hanjun Dai.
\newblock Annealed training for combinatorial optimization on graphs, 2022.
\newblock URL: \url{https://arxiv.org/abs/2207.11542}, \href
  {http://arxiv.org/abs/2207.11542} {\path{arXiv:2207.11542}}, \href
  {https://doi.org/10.48550/arXiv.2207.11542}
  {\path{doi:10.48550/arXiv.2207.11542}}.

\bibitem{SunY23}
Zhiqing Sun and Yiming Yang.
\newblock {DIFUSCO:} graph-based diffusion solvers for combinatorial
  optimization.
\newblock In Alice Oh, Tristan Naumann, Amir Globerson, Kate Saenko, Moritz
  Hardt, and Sergey Levine, editors, {\em Advances in Neural Information
  Processing Systems 36: Annual Conference on Neural Information Processing
  Systems 2023, NeurIPS 2023, New Orleans, LA, USA, December 10 - 16, 2023},
  2023.
\newblock URL:
  \url{http://papers.nips.cc/paper\_files/paper/2023/hash/0ba520d93c3df592c83a611961314c98-Abstract-Conference.html}.

\bibitem{Vishwanathan04}
Sundar Vishwanathan.
\newblock An approximation algorithm for finding long paths in hamiltonian
  graphs.
\newblock {\em J. Algorithms}, 50(2):246--256, 2004.
\newblock \href {https://doi.org/10.1016/S0196-6774(03)00093-2}
  {\path{doi:10.1016/S0196-6774(03)00093-2}}.

\bibitem{wang2022principled}
Haoyu Wang, Nan Wu, Hang Yang, Cong Hao, and Pan Li.
\newblock Unsupervised learning for combinatorial optimization with principled
  objective relaxation.
\newblock In {\em Advances in Neural Information Processing Systems},
  volume~35, pages 31444--31458, 2022.
\newblock URL:
  \url{https://proceedings.neurips.cc/paper_files/paper/2022/hash/cbc1ad2066f0afebbcea930c5688fc1f-Abstract-Conference.html}.

\bibitem{Zuckerman07}
David Zuckerman.
\newblock Linear degree extractors and the inapproximability of max clique and
  chromatic number.
\newblock {\em Theory Comput.}, 3(1):103--128, 2007.
\newblock URL: \url{https://doi.org/10.4086/toc.2007.v003a006}, \href
  {https://doi.org/10.4086/TOC.2007.V003A006}
  {\path{doi:10.4086/TOC.2007.V003A006}}.

\bibitem{Ore60}
Øystein Ore.
\newblock Note on hamilton circuits.
\newblock {\em The American Mathematical Monthly}, 67(1):55--55, 1960.
\newblock URL: \url{http://www.jstor.org/stable/2308928}.

\end{thebibliography}
\end{document}